\documentclass[SIADSpaper,onefignum,onetabnum]{siamonline250211}

\usepackage{lipsum}
\usepackage{amsfonts}
\usepackage{graphicx}
\usepackage{epstopdf}
\usepackage{algorithmic}
\usepackage{commath}
\usepackage{amsmath}
\usepackage{mathtools}
\usepackage[section]{placeins}
\usepackage{cleveref}
\usepackage[normalem]{ulem}
\usepackage{tikz}
\usepackage{nicematrix}
\usepackage{subcaption}
\ifpdf
  \DeclareGraphicsExtensions{.eps,.pdf,.png,.jpg}
\else
  \DeclareGraphicsExtensions{.eps}
\fi

\usepackage{enumitem}
\setlist[enumerate]{leftmargin=.5in}
\setlist[itemize]{leftmargin=.5in}

\newsiamremark{remark}{Remark}
\newsiamremark{hypothesis}{Hypothesis}
\crefname{hypothesis}{Hypothesis}{Hypotheses}
\newsiamthm{claim}{Claim}
\newsiamremark{fact}{Fact}
\crefname{fact}{Fact}{Facts}

\headers{Exact Finite Koopman Embedding of Block-Oriented Polynomial Systems}{Lucian Cristian Iacob, Roland T\'{o}th, and Maarten Schoukens}

\title{Exact Finite Koopman Embedding of Block-Oriented Polynomial Systems\thanks{Submitted to the editors \textbf{on the 12\textsuperscript{th} of July, 2025}.
\funding{This work was funded by the European Union (ERC, COMPLETE, 101075836). This research was also supported by the European Union within the framework of the National Laboratory for Autonomous Systems (RRF-2.3.1- 21-2022-00002). Views and opinions expressed are however those of the author(s) only and do not necessarily reflect those of the European Union or the European Research Council Executive Agency. Neither the European Union nor the granting authority can be held responsible for them.}}}

\author{Lucian Cristian Iacob\thanks{Control System Group, Dept.~of Electrical Engineering, Eindhoven Technical University, The Netherlands 
  (\email{l.c.iacob@tue.nl}, \email{r.toth@tue.nl}, \email{m.schoukens@tue.nl}).}
  \and Roland T\'oth\footnotemark[2]\hskip 1.5mm \textsuperscript{,}\thanks{Systems and Control Laboratory, 
HUN-REN Institute for Computer Science and Control, Hungary  
(\email{toth.roland@sztaki.hun-ren.hu}).}
\and Maarten Schoukens\footnotemark[2]}

\usepackage{amsopn}

\ifpdf
\hypersetup{
  pdftitle={Exact Finite Koopman Embedding of Block-Oriented Polynomial Systems},
  pdfauthor={Lucian Cristian Iacob, Roland T\'{o}th, and Maarten Schoukens}
}
\fi

\graphicspath{{./figures/}}

\newcommand{\Cp}{C}  
\newcommand{\Ap}{A}  
\newcommand{\Ry}{\bar{R}} 
\newcommand{\Rx}{R}       
\newcommand{\Py}{\bar{L}} 
\newcommand{\Px}{L}       

\newcommand{\Bz}{\bar{B}}       
\newcommand{\Bzct}{B}

\newcommand{\Cl}{\mathsf{C}}
\newcommand{\Al}{\mathsf{A}}
\newcommand{\Bl}{\mathsf{B}}
\newcommand{\Dl}{\mathsf{D}}

\newcommand{\mr}[1]{\mathrm{#1}}

\theoremstyle{definition}

\begin{document}

\maketitle

\begin{abstract}
The challenge of finding exact and finite-dimensional Koopman embeddings of nonlinear systems has been largely circumvented by employing data-driven techniques to learn models of different complexities (e.g., linear, bilinear, input affine). Although these models may provide good accuracy, selecting the model structure and dimension is still ad-hoc and it is difficult to quantify the error that is introduced. In contrast to the general trend of data-driven learning, in this paper, we develop a systematic technique for nonlinear systems that produces a finite-dimensional and exact embedding. If the nonlinear system is represented as a network of series and parallel linear and nonlinear (polynomial) blocks, one can derive an associated Koopman model that has constant state and output matrices and the input influence is polynomial. Furthermore, if the linear blocks do not have feedthrough, the Koopman representation simplifies to a bilinear model. 
\end{abstract}



\section{Introduction}
Developments in modern engineering increasingly rely on precise and comprehensive system modeling, while the drive to push and exceed technological limits has subsequently increased the performance requirements. As a result, models of systems have become ever more complex with nonlinearities that need to be considered to capture the full dynamic behavior. Although local linearization-based system analysis and control design methods have been available for a long time to handle nonlinear dynamics, they no longer provide the required performance and accuracy. This led to the development of various nonlinear control methods (e.g., dynamic programming, backstepping, feedback linearization, contraction, etc., \cite{khalil_nl}, \cite{Tsukamoto_contraction}, \cite{Zhou_HDP}), however, many of the existing results focus only on stability guarantees, are computationally complex, and performance of the resulting controller is difficult to shape. As such, recent years have seen a surge of research effort in embedding nonlinear systems into linear models, to make use of strong and well-developed control tools available for \emph{linear time-invariant} (LTI) systems. Some of these approaches are based on \emph{linear parameter-varying} (LPV) and \emph{linear time-varying} (LTV) modelling \cite{Mohammadpour}, \cite{LPV_Toth}, switched linear systems \cite{Switched_linear}, immersion \cite{Isidori}, \cite{Immersion_output}, or Carleman linearization \cite{Carleman_Original}, \cite{Carleman_bilinear}. A prevalent approach among the candidates is the Koopman framework, where the dynamics of the original system are lifted via (nonlinear) observable functions to a higher, possibly infinite-dimensional space, where the dynamics are linear and can be described via the so-called Koopman operator \cite{brunton_overview}, \cite{book_koopman}. Although such a linear embedding of autonomous systems is possible if an invariant Koopman subspace exists \cite{Brunton_2016}, only recently has it been shown that in the presence of external (control) inputs, the Koopman embedding results in at least bilinear dynamics on the input side, and in some cases even more complex input dynamics can occur; see \cite{bevanda2025}, \cite{brunton_overview}, \cite{Aut_Iacob_inputs}, \cite{Kaiser:20}.  \par

When it comes to practical applications, a major shortcoming of the Koopman framework is that there is no clear understanding whether a particular system can be embedded into an exact and most importantly finite-dimensional Koopman model. In absence of a solid theory, quite often data-driven methods are employed to identify the lifted model from data, often with surprisingly good accuracy \cite{iacob_id_j}, \cite{Lian_gp}, \cite{lusch_nn}, \cite{otto_bilinear}, \cite{williams_edmd}, \cite{kernel_edmd}. However, the resulting models are still inherently only approximations of the original system, and their representation capability depends on the available data, choice of model structure, or even lifting dimension \cite{Bruder}, \cite{Iacob_optimal}, \cite{Aut_Iacob_inputs}, \cite{Schulze:22}. 
Hence, one cannot hope to provide analysis guarantees or to design controllers with pre-described performance based on these models if there is no reliable characterization of the approximation error of the entire system dynamics (both the autonomous and input parts). Although recent research efforts aim to come up with reliable error bounds or uncertainty characterization for the obtained Koopman model to robustify the subsequent analysis and control design steps \cite{Error_bounds_Nuske}, \cite{Error_bounds_Philipp}, it is still a pending question under which conditions an exact finite-dimensional Koopman model of the system does exist in general and how we can calculate it in a computationally efficient manner.     
\par 
So far, useful, yet limited results have been obtained on the existence of finite dimensional Koopman-type embeddings for various system classes in terms of immersion or polyflows \cite{Isidori}, \cite{Polyflows}, \cite{levine_immersion}, \cite{liu_ozay_sontag_immersion}, which are based on recurrent Lie derivatives of the output and state, respectively. While these approaches provide interesting conditions to decide when the nonlinear system is 'embeddable', these conditions depend on checking whether the $n\textsuperscript{th}$ lifting function can be written as a linear combination of the previous $n-1$ functions for which no computational algorithm is known, making the testing of these conditions and the computation of the exact models difficult, resorting in many cases again to data-based approximations. 
Alternatively, Carleman linearization provides a constructive and computationally applicable method for computing Koopman-type models \cite{Carleman_Amini}, \cite{Carleman_Original}, \cite{Carleman_Hashemian}, \cite{Carleman_Rauh}, \cite{Carleman_bilinear}, however, it is difficult to decide when to stop with the linearization and extract, if it exists, an exact finite-dimensional form of the model representation.   
\par

In this paper, we aim to overcome this challenge by proposing a novel approach and a computable algorithm to construct exact finite-dimensional embeddings of nonlinear systems with inputs. The procedure that we propose focuses on embedding of nonlinear systems that are described via a network of elementary dynamic linear and static nonlinear blocks called \emph{block-oriented} or \emph{block-chain} nonlinear models. This class of systems is well known and is intensively used in many scientific fields such as filtering \cite{Pan:11}, \cite{Sasai:20}, robotics \cite{krzyzak_id}, biomedical applications \cite{hunt_muscle}, data-driven modeling and system identification \cite{book_block_id}, \cite{Schoukens_block_id}, etc.,  and well-known examples that fall into this class are series and parallel Wiener and Hammerstein models \cite{Schoukens_block_id}, \cite{WH_Parallel}, \cite{HW_Wills}. In our work, we consider the static nonlinear blocks to be multidimensional polynomials as many nonlinear functions have a convergent power series representation, see \cite{Alpay}, \cite{Beck}, hence a wide range of functions can be arbitrarily well represented by truncated power series, corresponding to finite-order polynomials. For a nonlinear system that is exactly represented by a nonlinear block-oriented model, we show in this paper that the dynamics of the original system can be exactly embedded into a finite-dimensional Koopman model. Furthermore, we provide an algorithm to compute this exact embedding.  

Some parallels between the proposed approach and Carleman linearisation can be drawn in the sense of taking time derivatives of Kronecker products of the state, however, the Carleman method takes an infinite linear combination of all possible monomials of the state with powers growing to infinity. Hence, while truncation of the 
Carleman linearization over a polynomial vector field  \cite{Carleman_Original}, \cite{Carleman_bilinear} turns out to be only an approximation of the nonlinear system, our method provides approximation-free embeddings of polynomial systems with block-chain representation. A connection of the present approach could also be made with \cite{Markovski_Wiener}, which embeds an autonomous Wiener model into an exact LTI model, but our methodology is capable of handling systems with input, extending the embedding to a wider range of systems at the expense of full linearity of the Koopman model. In fact, we show that the considered class of block-chain nonlinear systems have a \emph{polynomial input time-invariant} (PITI) Koopman models, which in case of no feedthrough in the dynamic linear blocks, simplifies to a \emph{bilinear time-invariant} (BLTI) Koopman representation. We summarize the contributions as follows:
\begin{itemize}
    \item Showing that block-chain polynomial systems without feedback element can always be embedded into the solution set of a PITI Koopman representation.
    \item We give conditions when the resulting PITI models are guaranteed to simplify to exact BLTI models.
    \item We provide a constructive iterative algorithm that, by iteratively processing the blocks of the block-chain nonlinear system, computes a finite-dimensional PITI Koopman form.
    \item We provide illustrative examples to showcase the applicability and validity of the algorithm.
\end{itemize} 
\par 
The paper is organized as follows. The preliminaries and the problem setting are given in  \Cref{sec:pre}. The main results on the existence of the finite-dimensional embedding are described in
\Cref{sec:block_embedding} together with an algorithm to compute the finite exact Koopman form. Finally, numerical examples are given in \Cref{sec:examples} and the conclusions are provided in
\Cref{sec:conclusions}.
\section{Preliminaries and Problem Setting} \label{sec:pre}
First, we discuss some preliminaries for Koopman embedding of autonomous systems and systems with inputs together with the considered problem setting of computing finite-dimensional exact Koopman embeddings of such systems. Then, the class of block-oriented polynomial nonlinear systems is introduced for which we aim to solve the finite-dimensional exact embedding problem.
\subsection{Koopman embedding of autonomous systems}
Consider a \emph{continuous-time} (CT) nonlinear system, given by the \emph{state-space} (SS) representation
\begin{subequations}\label{eq:nl_orig}
\begin{align}
	\dot{x}_t &= f(x_t),\label{eq:nl_orig:state}\\
	y_t &= h(x_t),\label{eq:nl_orig:output}
	\end{align}
\end{subequations}
where $x_t\in\mathbb{X} \subseteq \mathbb{R}^{n_\mathrm{x}}$ 
is the state, $y_t\in\mathbb{R}^{n_\mathrm{y}}$ is the output signal, $f:\mathbb{X} \rightarrow \mathbb{R}^{n_\mathrm{x}}$ and $h:\mathbb{X} \rightarrow \mathbb{R}^{n_\mathrm{y}}$ are the state and output functions, and $f$ is Lipschitz continuous, therefore the solutions of \eqref{eq:nl_orig} exist and are unique. In the Koopman framework, the nonlinear dynamics associated with the state $x_t$ is embedded into a linear dynamical relationship in a higher-dimensional space characterized by observables $\phi:\mathbb{X}\rightarrow\mathbb{R}$. 
These observables $\phi:\mathbb{X}\rightarrow \mathbb{R}$ are scalar functions (generally nonlinear) and are from a Banach function space $\mathcal{F}\subseteq  \mathcal{C}^1(\mathbb{X})$ with $\mathcal{C}^1(\mathbb{X})$ corresponding to continuously differentiable functions over $\mathbb{X}$. 

For \eqref{eq:nl_orig:state}, the solution $x_t$ is defined through the induced flow:
\begin{equation}\label{eq:flow_nl_ct}
x_t= F(t,x_0)=x_0 + \int^{t}_{0}f(x_\tau)\dif \tau.
\end{equation}
 The Koopman family of operators $\{\mathcal{K}^t:\mathcal{F}\rightarrow\mathcal{F}\}_{t\geq 0}$, associated with $F(t,\cdot)$, is defined by:
\begin{equation}
    \mathcal{K}^t\phi(x_0)=\phi\circ F(t,x_0), \quad \forall \phi\in\mathcal{F},
\end{equation}
where $\circ$ denotes function composition and the set $\mathbb{X}$ is considered to be open and forward invariant under $F(t,\cdot)$, i.e., $F(t,\mathbb{X})\subseteq\mathbb{X}$, $\forall t\geq 0$. Then, assuming that the Koopman semigroup of operators is strongly continuous \cite{book_koopman}, the infinitesimal generator of $\{\mathcal{K}^t\}_{t\geq 0}$, $\mathcal{L}:\mathcal{D}_\mathcal{L}$, is defined as:
\begin{equation}\label{eq:lim_generator}
\mathcal{L}\phi(x_0)=\lim_{t\downarrow 0}\frac{\mathcal{K}^t\phi(x_0)-\phi(x_0)}{t},\;\; \forall\phi\in\mathcal{D}_\mathcal{L},
\end{equation} 
with $\mathcal{D}_\mathcal{L}$ being a dense set in $\mathcal{F}$ and  the limit existing in a strong sense (see \cite{Lasota:94,book_koopman}). This means that, effectively, the Koopman generator can be used to describe the dynamics of the observables $\phi(\cdot)$ as:
\begin{equation}\label{eq:obs_generator}
\dot{\phi}=\frac{\partial \phi}{\partial x}f=\mathcal{L}\phi, \vspace{-.2cm}
\end{equation}
which is a linear representation of \eqref{eq:nl_orig:state}, albeit infinite dimensional in general. In practical applications, the embedding of \eqref{eq:nl_orig:state} into a finite-dimensional representation is often sought. This corresponds to a search for basis functions $\Phi^\top = [\begin{array}{ccc}
    \phi_1 & \cdots & \phi_{n_\mr f}
\end{array}]\in\mathcal{F}_{n_{\mr f}}$ such that $\mathcal{F}_{n_\mr f}$ is invariant under $\mathcal{L}$. Hence, due to the linearity of $\mathcal{L}$, we can write:
\begin{equation}\label{eq:generator_action_per_obs}
\dot{\phi}_j=\mathcal{L} \phi_j=\sum^{n_\mathrm{f}}_{i=1}L_{i,j}\phi_i,
\end{equation}
where $\mathcal{L}:\mathcal{F}_{n_\mr f}\rightarrow\mathcal{F}_{n_{\mr f}}$ and $\mathcal{F}_{n _\mr f}\subseteq \mathcal{D}_{\mathcal{L}}$. Here, $L$ denotes the matrix representation of the Koopman generator, and its $j^\textsuperscript{th}$ column contains the coordinates of $\mathcal{L}\phi_j$ expressed in the basis $\Phi$. Setting $A=L^\top$, we can express \eqref{eq:generator_action_per_obs} in a compact form as:
\begin{equation}\label{eq:koopman_aut}
\dot{\Phi}(x_t)=A\Phi(x_t).
\end{equation}
While \eqref{eq:koopman_aut} is often used to identify the Koopman dynamics (e.g. \cite{Klus}), it is generally solved only in an approximative sense. Outside of the Koopman literature there are methods to find an exact finite dimensional linear embedding  (see \cite{Isidori}, \cite{Polyflows}), that give conditions for \eqref{eq:koopman_aut} to exist under certain basis $\Phi$,
but no algorithm is provided to check if an exact embedding is practically possible, and resulting models are usually only approximations based on a heuristic choice of $n_\mr f$.

Using \eqref{eq:obs_generator}, the following relation also holds true:
\begin{equation}\label{eq:ct_koop_aut2}
\dot{\Phi}(x_t)=\frac{\partial \Phi}{\partial x}(x_t)f(x_t).
\end{equation}
Thus, to obtain a finite-dimensional Koopman embedding (i.e., lifting) for \eqref{eq:nl_orig:state}, the general requirement is finding a set of observables $\Phi$ such that:
\begin{equation}\label{eq:span_condition}
\frac{\partial \Phi}{\partial x}f\in \text{span} \left\lbrace\Phi\right\rbrace. 
 \end{equation}

Generally, the output map \eqref{eq:nl_orig:output} w.r.t.~the resulting embedding is defined as $h(x_t) = \Psi(\Phi(x_t))$, with $\Psi:\mathbb{R}^{n_\mathrm{f}}\rightarrow\mathbb{R}^{n_\mathrm{y}}$ a potentially nonlinear mapping. In this work, we will investigate existence of finite dimensional Koopman embeddings under the additional condition $h\in\text{span}\{\Phi\}$, allowing the output map to be written as:
\begin{equation}
	y_t = C\Phi(x_t).
\end{equation}
with $C\in\mathbb{R}^{n_\mathrm{y}\times n_\mathrm{f}}$. Note that this is not a limiting condition, as the class of systems considered in this paper directly satisfies this condition.

If a finite dimensional Koopman embedding of \eqref{eq:nl_orig} exists under the above considered conditions, then, by introducing $z_t=\Phi(x_t)$, we can write an equivalent state-space representation of \eqref{eq:nl_orig} as 
\begin{subequations}
\begin{align}
	\dot{z}_t &= Az_t,\\
	y_t&=Cz_t,
	\end{align}
\end{subequations}
with $z_0 = \Phi(x_0)$. 
\subsection{Koopman embedding of systems with inputs}
While the Koopman embedding of autonomous systems has been found to be rather powerful in describing complex fluid dynamics, in many engineering applications, systems are also affected by external inputs that influence the underlying system behavior. To be able to handle embedding under the presence of inputs, we consider general nonlinear systems described by a state-space representation:
\begin{subequations}\label{eq:nonlinear_sys_with_input}
\begin{align}
	\dot{x}_t &= f(x_t,u_t),\\
	y_t&= h(x_t,u_t),
	\end{align}
\end{subequations}
with $x_t \in \mathbb{X}\subseteq \mathbb{R}^{n_\mathrm{x}}$, $u_t\in\mathbb{U}\subseteq\mathbb{R}^{n_\mathrm{u}}$, and $f:\mathbb{R}^{n_\mathrm{x}}\times\mathbb{U}\rightarrow \mathbb{R}^{n_\mathrm{x}}$ being Lipschitz continuous. It is assumed that $\mathbb{U}$ is given such that $\mathbb{X}$ is open and forward invariant under the induced flow.

To obtain a Koopman embedding of \eqref{eq:nonlinear_sys_with_input}, as described in \cite{Aut_Iacob_inputs}, one can decompose $f(x_t,u_t)$ as follows:
\begin{equation}
	f(x_t,u_t)=f(x_t,0) + \underbrace{f(x_t,u)-f(x_t,0)}_{\bar{f}(x_t,u_t)}
\end{equation}
with $\bar{f}(x_t,0)=0$. Note that this decomposition always exists for any $f$, see \cite{surana_obs}, \cite{Aut_Iacob_inputs}. Next, we apply a similar decomposition to the output map:
\begin{equation}
	h(x_t,u_t) = h(x_t,0) + \underbrace{h(x_t,u_t)-h(x_t,0)}_{\bar{h}(x_t,u_t)}
\end{equation}
with $\bar{h}(x_t,0)=0$. Thus, the representation \eqref{eq:nonlinear_sys_with_input} becomes:
\begin{equation}
	\begin{split}
		\dot{x}_t&=f(x_t,0) + \bar{f}(x_t,u_t)\\
		y_t&= h(x_t,0) + \bar{h}(x_t,u_t).
	\end{split}
\end{equation}
Given a finite number of bases $\Phi$ such that condition \eqref{eq:span_condition} is satisfied for $f(x_t,0)$, then an exact Koopman representation of the dynamics is given by:
\begin{equation}\label{eq:lifted_dynamics_w_input}
	\dot{\Phi}(x_t)= A\Phi(x_t) + \mathcal{B}(x_t,u_t)u_t
\end{equation}
where, based on Lemma 1 in \cite{Aut_Iacob_inputs},
\begin{equation}
	\mathcal{B}(x_t,u_t)=\int^1_0\frac{\partial \Gamma}{\partial u}(x_t,\lambda u_t)\dif \lambda \quad \text{with} \quad \Gamma(x_t,u_t)=\frac{\partial \Phi}{\partial x}(x_t)\bar{f}(x_t,u_t).
\end{equation}
A similar procedure can be applied for the output map if  $h(x_t,0)\in\text{span}\{\Phi\}$. This condition can be easily satisfied by including the output in the dictionary of observables. Then, we obtain:
\begin{equation}
y_t = C\Phi(x_t) + \mathcal{D}(x_t,u_t)u_t, \quad \text{where} \quad \mathcal{D}(x_t,u_t) = \int^1_0\frac{\partial \bar{h}}{\partial u}(x_t,\lambda u_t)\dif \lambda.
\end{equation}
If a finite-dimensional Koopman embedding of \eqref{eq:nonlinear_sys_with_input} exists under the conditions $\frac{\partial \Phi}{\partial x}f(x_t,0)\in \text{span} \left\lbrace\Phi\right\rbrace $ and $h(x_t,0)\in\text{span}\{\Phi\}$, then, by introducing $z_t=\Phi(x_t)$, we can write an equivalent SS representation of \eqref{eq:nonlinear_sys_with_input} as 
\begin{subequations} \label{eq:koopman_sys_with_input}
\begin{align}
	\dot{z}_t &= Az_t + B(z_t,u_t)u_t,\\
	y_t&=Cz_t + D(z_t,u_t)u_t,
	\end{align}
\end{subequations}
with $z_0 = \Phi(x_0)$ under the assumption that there exist functions $ B$ and $ D$ such that the relations  $ B(\Phi(\cdot),\cdot)=\mathcal{B}(\cdot,\cdot)$ and $ D(\Phi(\cdot),\cdot)=\mathcal{D}(\cdot,\cdot)$ are satisfied. \par
Under certain conditions detailed in papers such as \cite{Goswami:17,Huang:18,Aut_Iacob_inputs,Schulze:22}, \eqref{eq:lifted_dynamics_w_input} can become bilinear, i.e., $B(z_t,u_t)$ reduces to affine dependency on $x_t$ only. If $\bar{f}(x_t,u_t)$ is input affine, i.e., $\bar{f}(x_t,u_t)=\tilde{f}(x_t)u_t$, then:
\begin{equation}
    \mathcal{B}(x_t,u_t)=\tilde{\mathcal{B}}(x_t)u_t, \quad \text{where} \quad 
    \tilde{\mathcal{B}}(x_t) = \frac{\partial \Phi}{\partial x}(x_t)\tilde{f}(x_t).
\end{equation}
As discussed in \cite{Aut_Iacob_inputs}, if 
\begin{equation}\label{eq:bilinear_state_condition}
    \frac{\partial \Phi}{\partial x}\tilde{f}_k\in\{\text{span}\{\Phi\}+\text{const}\}
\end{equation}
where $\tilde{f}_k$ is the $k^\textsuperscript{th}$ column of $\tilde{f}$, there exists a ${}_k\bar{B}\in\mathbb{R}^{n_\mr f \times n_\mr f}$ and ${}_kB \in\mathbb{R}^{n_\mr f \times 1}$ such 
that $\frac{\partial\Phi}{\partial x}\tilde{f}_k={}_k\bar{B}\Phi + {}_kB$. Note that the constant term also allows for fully LTI models to result from the embedding. Then, given that $\frac{\partial \Phi}{\partial x}f(x_t,0)\in \text{span} \left\lbrace\Phi\right\rbrace $, the lifted bilinear form of the dynamics is:
\begin{equation}\label{eq:koop_bilinear_form_in_u}
    \dot{\Phi}(x_t)=A\Phi(x_t) + \sum^{n_\mr u}_{k=1} \left({}_k\bar{B}\Phi(x_t)+{}_k B\right)u_{k,t},
\end{equation}
where $u_{k,t}$ is the $k\textsuperscript{th}$ element of $u_t$.
Similarly, for the output map, the necessary conditions can be described as follows. Let $\bar{h}(x_t,u_t)$ have an affine dependency on the input, i.e., $\bar{h}(x_t,u_t)=\tilde{h}(x_t)u_t$, such that the output function $h(x_t,u_t)$ is expressed as:
\begin{equation}
h(x_t,u_t) = h(x_t,0) + \tilde{h}(x_t)u_t.
\end{equation}
Then, if $\tilde{h}(x_t)\in\{\text{span}\{\Phi\}+\text{const}\}$ and $h(x_t,0)\in\text{span}\{\Phi\}$, the output equation can be written as: 
\begin{equation}
y_t=C\Phi(x_t)+\sum^{n_\mr u}_{k=1} \left({}_k\bar{D}\Phi(x_t)+{}_kD\right)u_{k,t}
\end{equation}
with ${}_k\bar{D}\in\mathbb{R}^{n_\mr f \times n_\mr f}$ and ${}_kD\in\mathbb{R}^{n_\mr f \times 1}$. Finally, let $z_t=\Phi(x_t)$, then the lifted exact finite-dimensional  Koopman form of \eqref{eq:nonlinear_sys_with_input} is given by:
\begin{subequations} \label{eq:BLTI_representation}
    \begin{align}
        \dot{z}_t &=Az_t + \sum^{n_\mr u}_{k=1} \left({}_k\bar{B} z_t+{}_k B\right)u_{k,t} =A z_t + \left(\sum^{n_{\mathrm{z}}}_{j=1}\bar{B}_{j} z_{j,t} + B\right)u_t \label{eq:BLTI_representation:state} \\
        y_t&=Cz_t+\sum^{n_\mr u}_{k=1} \left({}_k\bar{D}z_t+{}_kD\right)u_{k,t} = C z_t + \left(\sum^{n_{\mathrm{z}}}_{j=1}\bar{D}_{j} z_{j,t} + D\right)u_t \label{eq:BLTI_representation:output}
    \end{align}
\end{subequations}
with $z_0=\Phi (x_0)$, which corresponds to a \emph{bilinear time-invariant} (BLTI) system. Here, ${}_k{B}\in\mathbb{R}^{n_{\mathrm{z}}\times 1}$ gives the $k^{\text{th}}$ column of $B\in\mathbb{R}^{n_{\mathrm{z}}\times n_{\mathrm{u}}}$, while ${}_k\bar{B}\in\mathbb{R}^{n_{\mathrm{z}} \times n_{\mathrm{z}}}$  gives $\bar{B}_{j} = [\begin{array}{ccc}
	{}_{1,j} \bar{B} & \dots & {}_{n_{\mathrm{u}},j} \bar{B}
\end{array}]$ with ${}_{k,j} \bar{B}$ being the $j^{\text{th}}$ column of ${}_k \bar{B}$. The $D$ terms are similarly defined. 

Note that, in this paper, we will derive exact BLTI models where the output only depends on the lifted state, i.e., $y_t=Cz_t$. To increase readability and for the sake of simplicity, from here on we drop the subscript $t$ expressing time dependence.

\subsection{Block-oriented description of nonlinear systems}
To investigate when \eqref{eq:nonlinear_sys_with_input} can be converted to an exact Koopman form \eqref{eq:koopman_sys_with_input}, we restrict the scope of  considered systems to systems where the dynamics can be described by a block interconnection, in series and parallel, of LTI and static nonlinear blocks. The blocks are defined as follows: 
\subsubsection{LTI dynamic block} \label{sec:LTIblock}
The block $\Sigma^\mathrm{LTI}_i$ corresponds to an LTI system described by the minimal \emph{state-space} (SS) representation with dimensions $(n_{\mathrm{y},i},n_{\mathrm{x},i},n_{\mathrm{u},i})$:
\begin{subequations}\label{eq:LTI_block}
	\begin{align}
		\dot{x}_{i}&=\Al_ix_{i} + \Bl_iu_{i}= \Al_i x_{i} + \sum^{n_{\mathrm{u},i}}_{k=1}{}_k\Bl_i u_{i,k},\label{eq:LTI_block:state}\\
		y_{i}&=\Cl_ix_{i} + \Dl_iu_{i}= \Cl_i x_{i} + \sum^{n_{\mathrm{u},i}}_{k=1}{}_k\Dl_i u_{i,k},\label{eq:LTI_block:output}
	\end{align}
\end{subequations} 
where $x_{i,t}\in\mathbb{R}^{n_{\mathrm{x},i}}$ is the state of the representation, $u_{i,t}\in\mathbb{R}^{n_{\mathrm{u},i}}$ is the input of the block and $y_{i,t}\in\mathbb{R}^{n_{\mathrm{y},i}}$ is the output of the block. $\Al_i\in\mathbb{R}^{n_{\mathrm{x},i} \times n_{\mathrm{x},i}}$ is the state matrix, $\Bl_i\in\mathbb{R}^{n_{\mathrm{x},i}\times n_{\mathrm{u},i}}$ the input matrix with ${}_k\Bl_i$ being the $k\textsuperscript{th}$ column of $\Bl_i$, $\Cl_i\in\mathbb{R}^{n_{\mathrm{y},i}\times n_{\mathrm{x},i}}$ is the output matrix and $\Dl_i\in\mathbb{R}^{n_{\mathrm{y},i} \times n_{\mathrm{u},i}}$ is the feedthrough matrix with ${}_k\Dl_i$ being the $k\textsuperscript{th}$ column of $\Dl_i$.

At the level of an IO map, \eqref{eq:LTI_block} is expressed as
\begin{equation}
y_{i} = \underbrace{\left[\begin{array}{c|c} \Al_i & \Bl_i  \\ \hline \Cl_i & \Dl_i \end{array}  \right]}_{G_i} u_{i}, 
\end{equation}
where $G_i$ corresponds to an LTI operator whose Laplace transform is the transfer function $\Dl_i +\Cl_i(Is-\Al_i)^{-1}\Bl_i$ associated with \eqref{eq:LTI_block}, where $s\in\mathbb{C}$ is the complex frequency. 

\subsubsection{Static nonlinear block}\label{sec:nonlinear_block}
A nonlinear block $\Sigma^\mathrm{NL}_i$  with dimensions $(n_{\mathrm{y},i},n_{\mathrm{u},i})$ is described as:
\begin{equation}\label{eq:nl_block_def}
	y_{i} = f_i(u_{i})
\end{equation}
where $f_i:\mathbb{R}^{n_{\mathrm{u},i}}\rightarrow\mathbb{R}^{n_{\mathrm{y},i}}$ is a multivariate polynomial vector function. Note that many nonlinear functions have a convergent power series representation, see \cite{Alpay}, \cite{Beck}, hence all of these functions can be arbitrarily well represented by truncated power series, corresponding to a finite order polynomial. 

In order to embed a nonlinear static block \eqref{eq:nl_block_def}
into a Koopman form, $f$ is decomposed as a linear combination of univariate polynomials based on the approach in \cite{Dreesen}. For completeness, we give here a brief overview of the decomposition. For simplicity of the notation, we drop the subscript $i$ of $f$, then
the decomposition of $f(u)$ is written as:
\begin{equation} \label{eq:decomp}
	y=f(u)= Wg(V^\top u)
\end{equation}
where $V\in\mathbb{R}^{n_{\mathrm{u}}\times r}$, $W\in\mathbb{R}^{n_{\mathrm{y}}\times r}$. 
The function $g:\mathbb{R}^{r}\rightarrow\mathbb{R}^{r}$ is defined as 
\begin{equation}
g(V^\top u)=[\ g_{1}(\underbrace{v_1u}_{\sigma_1}) \ \cdots \ g_{r}(\underbrace{v_{r}u}_{\sigma_r}) \ ]^\top 
\end{equation}
with $g_{e}:\mathbb{R}\rightarrow\mathbb{R}$ 
being the scalar decoupled univariate polynomials, $v_e\in\mathbb{R}^{n_\mathrm{u}}$ being the $e^\mathrm{th}$ row of $V^\top$ and $\sigma=V^\top u$. The univariate scalar polynomials $g_e$ are defiend as:
\begin{equation}\label{eq:g_decouple_components}
    g_e(\sigma_e) = \gamma_{e,0} + \gamma_{e,1}\sigma_e + \dots + \gamma_{e,p}\sigma^p_e
\end{equation}
with $\{\gamma_{e,m}\}^p_{m=1}\in\mathbb{R}$, $\sigma_e$ being the $e\textsuperscript{th}$ element of $\sigma$, while $p$ represents the total degree of $f$ \cite{Dreesen}.
In \cite{Dreesen}, \cite{USEVICH2020}, it is shown that such a decomposition is possible for matrix polynomial functions $f$ given a sufficiently high $r\leq n_{\mathrm{y}}n_{\mathrm{u}}$.  For a given $r$, such decompositions can be computed by the toolbox \cite{decoup_toolbox}. 

To illustrate the decomposition mechanism, we provide a simple example. Let $y=[\begin{array}{cc} y_{1} & y_{2}\end{array}]^\top$, $u=[\begin{array}{cc} u_{1} & u_{2}\end{array}]^\top$ and $f=[\begin{array}{cc}
f_{1} & f_{2}
\end{array}]^\top$, such that $y=f(u)$ is written as:
\begin{equation} \label{exmp:decompostion}
\begin{bmatrix}
y_{1} \\ y_{2}
\end{bmatrix}=\begin{bmatrix}
f_{1}(u_{1},u_{2}) \\ f_{2}(u_{1},u_{2})
\end{bmatrix}=\begin{bmatrix}
u^2_{1}-4u_{1}u_{2}-2u_{1}+4u^2_{2} +4u_2+1 \\
-u^2_{1}+4u_{1}u_{2} + 2u_{1}+u_{2}^3 - 4u_{2}^2 - 5u_{2}-1
\end{bmatrix}
\end{equation}
It is possible to decompose this matrix polynomial with $r=2$:
\begin{equation}
\begin{bmatrix}
y_{1} \\ y_{2}
\end{bmatrix}=\underbrace{\begin{bmatrix}
1 & 0 \\ -1 & -1
\end{bmatrix}}_{W}\underbrace{\begin{bmatrix}
\sigma_{1}^2 - 2\sigma_{1} + 1 \\
\sigma_{2}^3 - \sigma_{2}
\end{bmatrix}}_{g(\sigma)}, \qquad \begin{bmatrix}
\sigma_{1} \\ \sigma_{2}
\end{bmatrix}= \begin{bmatrix}
v_{1}u \\ v_{2}u
\end{bmatrix} = \underbrace{\begin{bmatrix}
1 & -2 \\ 0 & -1
\end{bmatrix}}_{V^\top}\begin{bmatrix}
u_{1} \\ u_{2}
\end{bmatrix}.
\end{equation}
In this example $p=3$ and $r=2$. Writing $g(\sigma)$ in the form \eqref{eq:g_decouple_components}, we obtain the coefficients $\gamma$ as given in \Cref{tbl:table_example_g}.
\begin{table}\caption{Coefficients of $g$ resulting when the static polynomial block \eqref{exmp:decompostion} is decomposed into the form of \eqref{eq:decomp}.}\label{tbl:table_example_g}
\begin{center}
\begin{tabular}{ |c|c|c|c| } 
\hline
$\gamma_{1,0}=1$ & $\gamma_{1,1}=-2$ & $\gamma_{1,2}=1$ &  $\gamma_{1,3}=0$ \\ \hline
$\gamma_{2,0}=0$ & $\gamma_{2,1}=-1$ & $\gamma_{2,2}=0$ &  $\gamma_{2,3}=1$ \\ 
\hline
\end{tabular}
\end{center}
\end{table}

\subsubsection{Block-oriented nonlinear system representation}
Now we introduce a network representation of nonlinear systems in terms of interconnection of blocks, in series and parallel, of LTI and static nonlinear components. For this, we define a set of elementary block operations. These operations, performed iteratively from input to output, describe the dynamics of any series and parallel block interconnection of LTI and static nonlinear blocks. An example system is shown in \Cref{fig:nl_network_interconnection}. 
\begin{figure}[ht!]
\begin{center}    
\includegraphics[width=.99\textwidth]{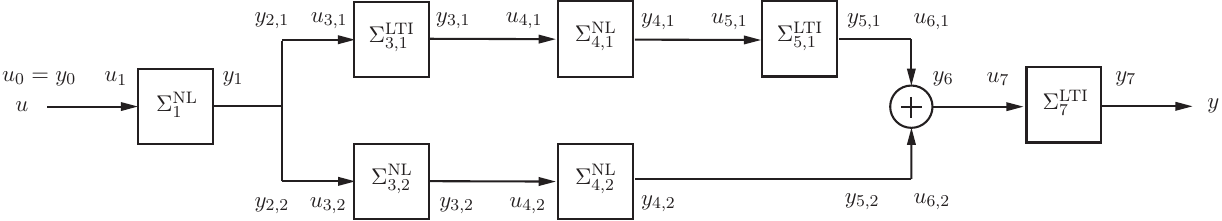}
\caption{Example of a block-chain interconnection of LTI blocks $\Sigma^\mathrm{LTI}_i$ and static nonlinear blocks $\Sigma^\mathrm{NL}_j$ in series and parallel.}
\label{fig:nl_network_interconnection}
\end{center}
\end{figure}
\FloatBarrier
The following operations with $i\in\mathbb{I}_{0}^N$, $N\geq 0$, are defined at the IO map level:
\begin{itemize}
\item \textit{Starting node}: The starting node is defined as $y_{0}=u$ with $y_{0,t}\in\mathbb{R}^{n_{\mathrm{y},0}}$ where $n_{\mathrm{y},0}=n_{\mathrm{u}}$. 
	\item \textit{Linear dynamic} (LD) block: Based on the LTI dynamics $\Sigma^\mathrm{LTI}_i$, represented by \eqref{eq:LTI_block}, \begin{equation}\label{eq:operation_1}
		y_{i} = G_i u_{i}, \quad \text{where }u_i=y_{i-1},
	\end{equation}
    with $y_{i,t}\in\mathbb{R}^{n_{\mathrm{y},i}}$ and $u_{i,t}\in\mathbb{R}^{n_{\mathrm{u},i}}$, where $n_{\mathrm{u},i}=n_{\mathrm{y},i-1}$.
	Note that $G_i$ can be both a dynamic operator defined by the matrices $(\Al_i,\Bl_i,\Cl_i,\Dl_i)$ or a static gain expressed by $\Dl_i$ only. 
	\item \textit{Static nonlinearity} (SN): Based on the NL map $\Sigma^\mathrm{NL}_i$, represented by  \eqref{eq:nl_block_def}, \begin{equation}\label{eq:operation_2}
		y_{i} = f_i(u_i), \quad \text{with }u_i=y_{i-1},
	\end{equation}
	where $y_{i,t}\in\mathbb{R}^{n_{\mathrm{y},i}}$ and $u_{i,t}\in\mathbb{R}^{n_{\mathrm{u},i}}$ with $n_{\mathrm{u},i}=n_{\mathrm{y},i-1}$. 
	\item  \textit{Input junction} (IJ): Corresponds to a branching of the signals
	\begin{equation}\label{eq:operation_3}
	\begin{split}
	y_{i,1}&=u_{i,1}=y_{i-1},\\
	&\;\;\vdots \\
	y_{i,m}&=u_{i,m}=y_{i-1},
	\end{split}
	\end{equation}
	for a junction of $m$ branches with $n_{\mathrm{y},i,j}=n_{\mathrm{y},i-1}$ for all $j\in \{1,\ldots,m\}$. Input junction is only possible if $N>1$, as it is required to be followed by an output junction somewhere in the block chain. Note that to avoid technical clutter, w.l.o.g.~we do not define signal splitting (multiplexing), i.e., a junction where $y_{i,j}=S_j y_{i-1}$ with $S_j\in \mathbb{I}^{n_{\mathrm{y},i,j} \times n_{\mathrm{y},i-1}}$ being a full-row rank selector matrix containing only $1$ and $0$ with $0<n_{\mathrm{y},i,j}\leq n_{\mathrm{y},i-1}$.
	\item \textit{Output junction (OJ):}  Corresponds to summing of the signals \begin{equation}\label{eq:operation_4}
		y_{i} = \sum^m_{j=1} u_{i,j} = \sum^m_{j=1} y_{i-1,j},
	\end{equation}
    for a junction of $m$ branches with $n_{\mathrm{y},i}=n_{\mathrm{y},i-1,j}$ for all $j\in \{1,\ldots,m\}$. Note that an OJ is only possible if it has been preceded by an IJ, i.e, there are branches to join. Again, to avoid technical clutter, w.l.o.g.~we do not define signal de-multiplexing, i.e., a junction where $u_{i,j}=S_j y_{i-1,j}$ with $S_j\in \mathbb{I}^{n_{\mathrm{y},i} \times n_{\mathrm{y},i-1,j}}$ being a full-column rank selector matrix containing only $1$ and $0$ with $0<n_{\mathrm{y},i-1,j}\leq n_{\mathrm{y},i}$. 
    \item\textit{End node:} Defined as $y=y_{N}$ with $N\in\mathbb{N}$ being the index of the last block-chain element preceding the end node and $y_{N,t}\in\mathbb{R}^{n_{\mathrm{y},N}}$ where $n_{\mathrm{y},N}=n_{\mathrm{y}}$. An end node is only possible if each IJ in the block chain has been closed by an OJ.
\end{itemize}
Note that cases of multiplexing and demultiplexing can be handled via zero padding of the corresponding signals. However, feedback interconnection is not considered in our block-oriented setting due to technical convenience to avoid problems of well-posedness and limitations of the conversion theory we present in Section \ref{sec:block_embedding}. Furthermore, for autonomous systems without inputs, the same block chain representation can be applied with minor adaptations of the starting node and the first element. 

Nevertheless, well-known NL model structures in the literature such as Wiener, Hammerstein, or subsequent combinations (e.g., \cite{WH_Parallel}, \cite{Schoukens_block_id}, \cite{HW_Wills}) can be easily represented as block-oriented models by the above-defined operations as exemplified in \Cref{fig:wiener_hammerstein}. However, the absence of a feedback operation means that Lur'e type of nonlinear systems fall out of the considered system setting.

\begin{figure}[ht!]
\begin{center}  
\includegraphics[width=.99\textwidth]{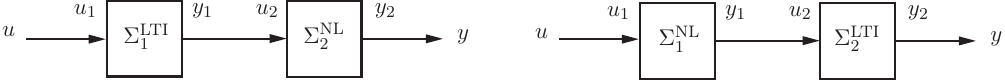}
\caption{Block oriented description of Wiener (left) and Hammerstein (right) systems.}
\label{fig:wiener_hammerstein}
\end{center}
\end{figure}
\FloatBarrier

\section{Finite Koopman embeddings of block-oriented NL representations}\label{sec:block_embedding} 
With all the preliminaries introduced, we are now ready to state our main result on the exact finite-dimensional Koopman embedding of NL systems that have a block-oriented representation with polynomial NL blocks. 
\subsection{Embedding theorems}
We begin by formulating a special case of the Koopman form \eqref{eq:koopman_sys_with_input} in terms of a \emph{polynomial input time-invariant} (PITI) Koopman form $\Sigma^\mathrm{PITI}$: 
\begin{subequations}\label{eq:PITI_representation}
	\begin{align}
		\dot{z} &= \Ap   z + \Px (z)\Rx (u)u,\\
		y &= \Cp z + \Py  (z)\Ry (u)u, \label{eq:PITI_representation:output}
	\end{align}
\end{subequations}
with lifted state $z_t=\Phi(x_t)\in\mathbb{R}^{n_{\mr z}}$ and state and output matrices $\Ap  \in\mathbb{R}^{n_{\mathrm{z}} \times n_{\mathrm{z}}}$, $\Cp\in\mathbb{R}^{n_{\mathrm{y}}\times n_{\mathrm{z}}}$. The functions $\Px :\mathbb{R}^{n_{\mathrm{z}}}\rightarrow\mathbb{R}^{n_{\mathrm{z}}\times n_{\mathrm{r}}}$ and $\Py :\mathbb{R}^{n_{\mathrm{z}}}\rightarrow\mathbb{R}^{n_{\mathrm{z}}\times n_{\bar{\mathrm{r}}}}$ are linear in $z$, while $\Rx :\mathbb{R}^{n_{\mathrm{u}}}\rightarrow\mathbb{R}^{n_{\mathrm{r}}\times n_{\mathrm{u}}}$ and $\Ry :\mathbb{R}^{n_{\mathrm{u}}}\rightarrow\mathbb{R}^{n_{\bar{\mathrm{r}}} \times n_{\mathrm{u}}}$ are polynomials in $u$. 

The following theorem holds:

\begin{theorem}\label{thm:PITI_theorem}
	Given a nonlinear system \eqref{eq:nonlinear_sys_with_input} whose dynamics can be represented as a block-chain of $\Sigma^{\mathrm{LTI}}$, see \eqref{eq:LTI_block}, and $\Sigma^{\mathrm{NL}}$ blocks, see \eqref{eq:nl_block_def}, in terms of the operations~\crefrange{eq:operation_1}{eq:operation_4}, then system \eqref{eq:nonlinear_sys_with_input} has an exact finite-dimensional PITI Koopman representation in the form of \eqref{eq:PITI_representation}.
\end{theorem}

Before proving \cref{thm:PITI_theorem}, the following result gives a simplification of it: 

\begin{corollary}\label{thm:bilinear_theorem}
	Given a nonlinear system \eqref{eq:nonlinear_sys_with_input} which, in terms of \cref{thm:PITI_theorem}, can be written in the PITI form of  \eqref{eq:PITI_representation}. If the following conditions are satisfied by the block-chain representation of \eqref{eq:nonlinear_sys_with_input}:
    \begin{enumerate}[label=(\roman*)]
        \item each $\Sigma^{\mathrm{LTI}}_i$ block has no feedthrough ($D_i=0_{{n_{\mathrm{y},i}\times n_{\mathrm{u},i}}}$), 
        \item the first operation following $y_0=u$ is not SN \eqref{eq:operation_2} or IJ \eqref{eq:operation_3} followed by SN,
    \end{enumerate}
   then \eqref{eq:PITI_representation} reduces to a BLTI Koopman representation \eqref{eq:BLTI_representation}.
   \end{corollary}

We will prove \Cref{thm:PITI_theorem} and \Cref{thm:bilinear_theorem} inductively in 
 \Cref{sec:proof_section} by first discussing the PITI Koopman embedding of elementary blocks and then showing that applying any interconnection operation of the block chain in relation with a PITI model will produce a PITI Koopman model of the joint dynamics. We will also show how each of the steps simplify to a BLTI form if the conditions of \Cref{thm:bilinear_theorem} are satisfied. Note that the step-by-step constructive proof also provides an algorithm to compute an exact finite dimensional Koopman embedding which is also a major contribution of the present paper.   
 
\subsection{Embedding an LD in PITI}\label{sec:embedding:LD} The block-chain representation can either start with an LD or an SN block or an IJ, hence, as a preparation for a formal proof of \Cref{thm:PITI_theorem} and \Cref{thm:bilinear_theorem}, we will first discuss the conversion of LD and SN blocks to a PITI/BLTI Koopman form, while we will handle IJs in a separate manner in \Cref{sec:inputjunc}. 

An LTI block $\Sigma_1^\mathrm{LTI}$ can be easily expressed in a PITI Koopman representation \eqref{eq:PITI_representation}, see \Cref{fig:LTI_to_PITI}, as follows.
\begin{figure}[ht!]
\begin{center}   
\includegraphics[scale=.95]{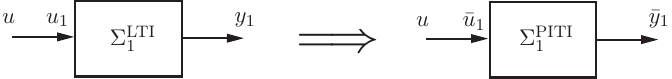}
\caption{Embedding an LTI block in a PITI Koopman representation. \vspace{-.5cm}}
\label{fig:LTI_to_PITI}
\end{center}
\end{figure}
\FloatBarrier
According to \Cref{sec:LTIblock}, $\Sigma_1^\mathrm{LTI}$ is given by
\begin{subequations}\label{eq:LTI_block_dyn}
	\begin{align}
	\dot{x}_1 &= \Al_1x_1 + \Bl_1u_1= \Al_1 x_1 + \sum^{n_{\mathrm{u},1}}_{k=1}{}_k\Bl_1 u_{1,k},\\
	y_1&=\Cl_1x_1 + \Dl_1 u_1.
	\end{align}
\end{subequations} 
Recall that ${}_k\Bl_1$ is the $k\textsuperscript{th}$ column of $\Bl_1$. Next, we give a lemma and a corollary for the PITI and BLTI formulations.
\begin{lemma}\label{lemma:LTI_OP1}
	A linear block $\Sigma_1^\mathrm{LTI}$ corresponding to \eqref{eq:operation_1} with an SS form \eqref{eq:LTI_block_dyn} can be written in  PITI form $\Sigma^\mathrm{PITI}_1$, given by \eqref{eq:PITI_representation}, with state $z_1=x_1$, input $\bar{u}_1=u_1$, output $\bar{y}_1=y_1$,  $\Ap _1=\Al_1$, $\Px _1(z_1)\equiv \Bl_1$, $\Rx _1(\bar{u}_1) \equiv I_{n_{\mathrm{u},1}}$, $\Cp _1=\Cl_1$, $\Py_1 (z_1) \equiv {\Dl_1}$, $\Ry _1(\bar{u}_1)={I_{n_{\mathrm{u},1}}} $  with $n_{\mathrm{r},1}=n_{\mathrm{u},1}$ and $n_{\bar{\mathrm{r}},1}=n_{\mathrm{y},1}$.
\end{lemma}
\begin{proof}
	By substitution of the above given matrices and functions into \eqref{eq:PITI_representation}, the result trivially follows. 
\end{proof}

\begin{corollary}\label{corr:LTI_OP1}
For a linear block $\Sigma_1^\mathrm{LTI}$, the resulting Koopman form by \Cref{lemma:LTI_OP1} is always a BLTI Koopman representation \eqref{eq:BLTI_representation}. If there is no feedforward term in $\Sigma_1^\mathrm{LTI}$ (i.e. $\Dl_1=0$), then the BLTI Koopman form also does not have a feedforward term. 
\end{corollary}
\begin{proof}
	It is simple to see that $\Px _1(z_1) \equiv  \Bl_1$, $\Rx _1(\bar{u}_1) \equiv  I_{n_{\mathrm{u},1}}$ implies that $\Px _1(z_1)\Rx _1(\bar{u}_1)\equiv \Bl_1$ and with ${}_kB_1={}_k\Bl_1$ and ${}_k\bar{B}_1=0$, one obtains \eqref{eq:BLTI_representation:state}. The output equation \eqref{eq:BLTI_representation:output} similarly follows. Furthermore, $\Py _1(z_1)\Ry _1(\bar{u}_1)\equiv 0$ if $\Dl_1=0_{n_{\mathrm{y},1}\times n_{\mathrm{u},1}}$, as $\Py _1(\bar{u}_1)\equiv \Dl_1$. 
\end{proof}
\subsection{Embedding an SN in PITI}\label{sec:input_nonlinearity} Next, we discuss embedding of a static nonlinear block into 
 a PITI form.
According to \Cref{sec:nonlinear_block}, a nonlinear block $\Sigma^\mathrm{NL}_1$ is described as:
\begin{equation} \label{eq:snblock}
	y_{1} = f_1(u_{1}).
\end{equation}
  The embedding into a PITI representation as shown in \Cref{fig:f_to_PITI}, is done through the conversion of \eqref{eq:snblock} into a state-space representation. The first step is to write the following trivial decomposition of $f_1$:
  \begin{equation}\label{eq:decomp_f1_sn_to_piti}
      f_1(u_1) = f_1(0) + \underbrace{f_1(u_1)-f_1(0)}_{\bar{f}_1(u_1)},
  \end{equation}
which always holds. Next, we use the exact factorization detailed in Lemma 1 in \cite{Aut_Iacob_inputs}, giving: \vspace{-.2cm}
\begin{equation}\label{eq:tilde_f1_sn_to_piti}
    \bar{f}_1(u_1) = \underbrace{\left(\int^1_0\frac{\partial \bar{f}_1}{\partial u_1}(\lambda u_1)\dif \lambda\right)}_{\tilde{f}_1(u_1)}u_1. \vspace{-.1cm}
\end{equation}
The resulting $\tilde{f}_1$ is polynomial in $u_1$. We can now formulate the embedding lemma. 
\begin{figure}[ht!]
\begin{center}   
\includegraphics[scale=0.95]{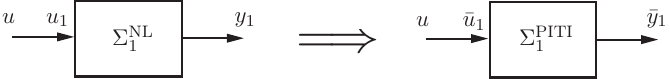}
\caption{Embedding a static nonlinear block in a PITI Koopman representation.\vspace{-.7cm}}
\label{fig:f_to_PITI}
\end{center}
\end{figure}
\begin{lemma}\label{lemma:SN_OP2}
    A static nonlinear block $\Sigma_1^\mathrm{NL}$ corresponding to \eqref{eq:snblock} that is decomposed as \eqref{eq:decomp_f1_sn_to_piti} with \eqref{eq:tilde_f1_sn_to_piti}, can be written in a PITI form \eqref{eq:PITI_representation}, with state $z_{1}\equiv 1\in\mathbb{R}^{n_{\mathrm{z},1}}$, $n_{\mathrm{z},1}=1$, input $\bar{u}_1=u_1$, output $\bar{y}_1=y_1$, $\Ap  _1 =0_{n_{\mathrm{z},1}\times n_{\mathrm{z},1} }$, $\Px _1(z_1)\equiv 0_{n_{\mathrm{z},1}\times n_{\mathrm{u},1} }$, $\Rx _1(\bar{u}_1) \equiv I_{n_{\mathrm{u},1}}$, $\Cp_1=f_1(0)$, $\Py _1(z_1) \equiv I_{n_{\mathrm{y},1}}$, $\Ry _1(\bar{u}_1) \equiv \tilde{f}_1(u_1)$.
\end{lemma}
\begin{proof}
	By substitution of the above given matrices and functions into \eqref{eq:PITI_representation}, the result trivially follows. 
\end{proof}

\par Note that, while the SN block can be described as a Koopman PITI model, it cannot be simplified to a BLTI Koopman representation with no feedtrough due to the presence of a polynomial feedthrough of $u_1$. As we will see later, we can only guarantee that the BLTI property of the Koopman model will be preserved by follow-up block absorptions into it, if the previous operations resulted in a BLTI Koopman model without feedtrough. Because of this, if the first block of the block chain is a static nonlinearity, then the overall resulting PITI Koopman model from the embedding might not be reducible to a BLTI one. If it is necessary to obtain a bilinear representation, one can choose to circumvent this input nonlinearity by constructing a virtual input as $\tilde{u}_1=f_1(u_1)$, however, certain utilization of the resulting model, e.g., for control design, becomes more complicated. 

\subsection{Embedding PITI followed by LD into PITI} \label{sec:embedding:PITI:LD}
This subsection details the conversion of a series interconnection between a PITI block $\Sigma_{i-1}^\mathrm{PITI}$ and an LTI block $\Sigma_{i}^\mathrm{LTI}$  into a single PITI Koopman model $\Sigma_i^\mathrm{PITI}$ for $i>1$. The interconnection is represented in \Cref{fig:PITI_LTI_to_PITI}.
\begin{figure}[ht!]
\begin{center}   
\includegraphics[width=.95\textwidth]{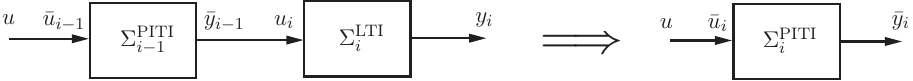}
\caption{Embedding the series interconnection of a PITI block and a linear dynamic block into a single PITI Koopman representation.\vspace{-.4cm}}
\label{fig:PITI_LTI_to_PITI}
\end{center}
\end{figure}
\FloatBarrier
The embedding is detailed in the following lemma.
\begin{lemma}\label{lemma:rule_PITI_LD}
	Series interconnection between a PITI block $\Sigma_{i-1}^\mathrm{PITI}$ and an LTI block $\Sigma_{i}^\mathrm{LTI}$  can be represented by an exact finite dimensional PITI Koopman representation $\Sigma_{i}^\mathrm{PITI}$  in  the form of \eqref{eq:PITI_representation} with state $z_i=[\begin{array}{ccc}
		z^\top_{i-1} & x^\top_i
	\end{array}]^\top$, input $\bar{u}_i=\bar{u}_{i-1}$, output $\bar{y}_i=y_i$, and 
\begin{subequations}\label{eq:PITI_LTI_matrices}
	\begin{align}
	\Ap  _i&= \begin{bmatrix}
		\Ap  _{i-1} & 0 \\ \Bl_i\Cp_{i-1} & \Al_i
	\end{bmatrix},& \Px _i(z_i) &\!=\! \begin{bmatrix}
		\Px _{i-1}(z_{i-1}) & 0 \\ 0 & \Bl_i\Py _{i-1}(z_{i-1})
	\end{bmatrix},& \Rx _i(\bar{u}_i)&\!=\!\begin{bmatrix}
		\Rx _{i-1}(\bar{u}_{i-1}) \\ \Ry _{i-1}(\bar{u}_{i-1})
	\end{bmatrix}\\
	\Cp_i &=\begin{bmatrix}
		\Dl_i\Cp_{i-1} & \Cl_i
	\end{bmatrix},& \Py _i(z_i)&\!=\! \Dl_i\Py _{i-1}(z_{i-1}),& \Ry _i(\bar{u}_i)&\!=\!\Ry _{i-1} (\bar{u}_{i-1}).
	\end{align}
\end{subequations}
\end{lemma}
\begin{proof}
	The proof follows by substituting \eqref{eq:PITI_representation:output} with output $\bar{y}_{i-1}$ into \eqref{eq:LTI_block:state} and \eqref{eq:LTI_block:output} under $u_i=\bar{y}_{i-1}$ and appending the state as $z_i=[\begin{array}{ccc}
		z^\top_{i-1} & x^\top_i
	\end{array}]^\top$. Note that, $\Px _i(z_i)$ is linear in $z_i$, because $\Px _{i-1}(z_{i-1})$ and $\Py _{i-1}(z_{i-1})$ are linear in $z_{i-1}$.
\end{proof}
\begin{corollary} \label{corr:rule_PITI_LD}
If the the PITI block $\Sigma_{i-1}^\mathrm{PITI}$  is bilinear, then the Koopman embedding  $\Sigma_{i}^\mathrm{PITI}$ resulting from \Cref{lemma:rule_PITI_LD} is bilinear and can be written in the form of \eqref{eq:BLTI_representation}. In case one of the blocks $\Sigma_{i-1}^\mathrm{PITI}$ or $\Sigma_{i}^\mathrm{LTI}$ has no feedtrough term, then $\Sigma_{i}^\mathrm{PITI}$ also has no feedtrough term, i.e., $\Py _{i}(z_{i})\Ry _{i}(\bar{u}_i)$ is zero.
\end{corollary}
\begin{proof}
    It is trivial to see that, if, due to bilinearity, $\Rx _{i-1}(\bar{u}_{i-1})$ and $\Ry _{i-1}(\bar{u}_{i-1})$ are stacks of identity matrices, i.e.,  $\Rx _{i-1}(\bar{u}_{i-1}) \equiv I_{n_{\mathrm{u},i-1}}$ and $\Ry _{i-1}(\bar{u}_{i-1}) \equiv I_{n_{\mathrm{u},i-1}}$, then $\Rx _{i}(\bar{u}_i) \equiv [\begin{array}{cc} I_{n_{\mathrm{u},i-1}} & I_{n_{\mathrm{u},i-1}} \end{array}]^\top$ and $\Ry _{i}(\bar{u}_i) \equiv I_{n_{\mathrm{u},i-1}} $ proving bilinearity of the resulting Koopman form. In case, either $\Dl_{i}=0_{n_{\mathrm{y},i}\times n_{\mathrm{u},i}}$ or the relation $\Py _{i-1}(z_{i-1})\Ry _{i-1}(\bar{u}_{i-1}) \equiv 0_{n_{\mathrm{y},i-1} \times n_{\mathrm{u},i-1}}$, then the direct feedtrough term given by $\Py _{i}(z_{i})\Ry _{i}(\bar{u}_i)=\Dl_{i}\Py _{i-1}(z_{i-1})\Ry _{i-1}(\bar{u}_{i-1})$ is zero.
\end{proof}

\subsection{Embedding PITI followed by SN into PITI} \label{sec:embedding:PITI:SN} This subsection details the conversion of a series interconnection between a PITI block $\Sigma_{i-1}^\mathrm{PITI}$ and an SN block $\Sigma_{i}^\mathrm{NL}$  into a single PITI Koopman model $\Sigma_i^\mathrm{PITI}$ for $i>1$. The interconnection is represented in \Cref{fig:PITI_NL_to_PITI}.
\begin{figure}[ht!]
\begin{center}   
\includegraphics[width=.95\textwidth]{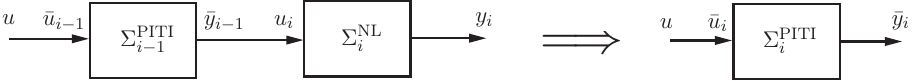}
\caption{Embedding the series interconnection of a PITI block and a static nonlinear block into a single PITI Koopman representation. 
}
\label{fig:PITI_NL_to_PITI}
\end{center}
\end{figure}
\FloatBarrier
As a result of the series interconnection we have 
\begin{subequations}\label{eq:PITI_f}
	\begin{align}
		\dot{z}_{i-1} &=\Ap  _{i-1} z_{i-1} + \Px _{i-1}(z_{i-1})\Rx _{i-1}(\bar{u}_{i-1})\bar{u}_{i-1},\\
		\bar{y}_{i-1} &= \Cp_{i-1} z_{i-1} + \Py _{i-1}(z_{i-1})\Ry _{i-1}(\bar{u}_{i-1})\bar{u}_{i-1},\\
		y_{i} &= f_i(\bar{y}_{i-1}) = W_ig_i(V_i^\top \bar{y}_{i-1})=[\ g_{i,1}(\underbrace{v_{i,1}\bar{y}_{i-1} }_{\sigma_{i,1}}) \ \cdots \ g_{i,r_i}(\underbrace{v_{i,r_i}\bar{y}_{i-1}}_{\sigma_{i,r_i}}) \ ]^\top . \label{eq:piti_f_output_nl}
	\end{align}
\end{subequations}
where, according to \Cref{sec:nonlinear_block}, $g_{i,e}:\mathbb{R}\rightarrow\mathbb{R}$ 
being scalar univariate polynomials, $v_{i,e}\in\mathbb{R}^{n_{\mathrm{y},i-1}}$ being the $e^\mathrm{th}$ row of $V_i^\top$ and $\sigma_i=V_i^\top \bar{y}_{i-1}$, while $z_{i,t}\in\mathbb{R}^{n_{\mathrm{z},i}}$. Furthermore,
\begin{equation}\label{eq:decomp_ge_piti_nl}
    g_{i,e}(\sigma_{i,e}) = \gamma_{i,e,0} + \gamma_{i,e,1}\sigma_{i,e} + \dots + \gamma_{i,e,p_i}\sigma^{p_i}_{i,e}
\end{equation}
with $\{\gamma_{i,e,m}\}^{p_i}_{m=1}\in\mathbb{R}$ and  $\sigma_{i,e}$ being the $e\textsuperscript{th}$ element of $\sigma_i$. Furthermore, let $x^{(\tau)}$ denote the $\tau^\mathrm{th}$ Kronecker power of a $x\in\mathbb{R}^{n_\mathrm{x}}$:
\begin{equation*}
x^{(0)} = 1,\quad x^{(1)} = x,\quad x^{(2)} = x \otimes x,\quad \ldots\quad x^{(\tau)} = \underbrace{x \otimes x \otimes \cdots \otimes x}_{\tau \text{\ times}}. 
\end{equation*}
Then, the embedding is detailed in the following lemma.
\begin{lemma}\label{lemma:rule_PITI_SN}
	The series interconnection between a PITI block $\Sigma_{i-1}^\mathrm{PITI}$ and an SN block $\Sigma_{i}^\mathrm{NL}$  can be represented by an exact finite dimensional PITI Koopman representation $\Sigma_{i}^\mathrm{PITI}$  in  the form of \eqref{eq:PITI_representation} with state $z_i=[\begin{array}{cccc}
		1 & z^\top_{i-1} & \cdots & (z^{(p_i)}_{i-1})^\top
	\end{array}]^\top$,   
    input $\bar{u}_i=\bar{u}_{i-1}$, output $\bar{y}_i=y_i$, and state equation defined by
    \begin{subequations}
	\begin{align}
	\Ap  _i &\!=\!\! \begin{bmatrix}
		0 & 0 & \cdots & 0 \\
		0 & \Ap  _{i-1} & \ddots & 0 \\
		\vdots & \ddots & \ddots & 0 \\
		0 & \cdots &0  & {}^{p_i}\Ap  _{i-1}
	\end{bmatrix}\!, & \!\!\!
	\Px _i(z_i) &\!=\!\! \begin{bmatrix}
		0_{1\times n_{\mr z , i-1}}
        \\ I_{{n_{\mr z,{i-1}}}} \\ \frac{\partial z^{(2)}_{i-1}}{z_{i-1}} \\ \vdots \\ \frac{\partial z^{(p_i)}_{i-1}}{\partial z_{i-1}}
	\end{bmatrix}\!\Px _{i-1}(z_{i-1}), &\!\!\! \Rx _i(\bar{u}_i) &\!\equiv\! \Rx _{i-1}(\bar{u}_{i-1}),
\end{align}
where  ${}^{\tau}A_{i-1}=\sum^{\tau-1}_{k=0}I_{n_{\mathrm{z},i-1}}^{(k)} \otimes  A_{i-1} \otimes I_{n_{\mathrm{z},i-1}}^{(\tau-k-1)}$, with $\tau\in\{2,\dots,p_i\}$, and $\frac{\partial z^{(\tau)}_{i-1}}{\partial z_{i-1}}$ is defined in terms of \Cref{lemma:appendix_lemma_1} in \Cref{app:A2}.
The output equation is defined by
\begin{align}
\Cp_i &= W_i\Gamma_i, &
	\Py _i(z_i) &= W_i\begin{bmatrix}
		\Py _{i,1}(z_{i-1}) \\
		\vdots \\
		\Py _{i,r_i}(z_{i-1})
	\end{bmatrix} & \Ry _i(\bar{u}_i) &= \begin{bmatrix}
		\Ry _{i-1}(\bar{u}_{i-1}) \\
		\Ry _{i-1}(\bar{u}_{i-1}) \\
		\Ry _{i-1}^{(2)}(\bar{u}_{i-1})\left(I_{n_{\mathrm{u},{i-1}}}\otimes \bar{u}_{i-1}\right)\\
		\Ry _{i-1}(\bar{u}_{i-1}) \\
		\Ry _{i-1}^{(2)}(\bar{u}_{i-1})\left(I_{n_{\mathrm{u},{i-1}}}\otimes \bar{u}_{i-1}\right)\\
		\Ry _{i-1}^{(3)}(\bar{u}_{i-1})\left(I_{n_{\mathrm{u},{i-1}}}\otimes \bar{u}_{i-1}^{(2)}\right)\\
		\vdots \\
		\Ry _{i-1}^{(p_i)}(\bar{u}_{i-1})\left(I_{n_{\mathrm{u},{i-1}}}\otimes \bar{u}_{i-1}^{(p_i-1)}\right)
	\end{bmatrix}
\end{align}
\end{subequations}
with \begin{align}
 \Gamma_i = \begin{bmatrix}
		\Gamma^\top_{i,1} & \dots & \Gamma^\top_{i, r_i}
	\end{bmatrix}^\top  \text{ and }\; \Gamma_{i,e}=\begin{bmatrix}
		\gamma_{i,e,0} \; \gamma_{i,e,1} \tilde{v}_{i,e} \; \gamma_{i,e,2}\tilde{v}^{(2)}_{i,e} \; \cdots \; \gamma_{i,e,p_i}\tilde{v}^{(p_i)}_e
	\end{bmatrix} ,
\end{align}
\begin{multline} \label{eq:piti_nl_barL}
	\Py _{i,e}(z_{i-1})=\left[ \begin{matrix}
		\gamma_{i,e,1}\tilde{v}^{(0)}_{i,e}z^{(0)}_{i-1}v^{(1)}_{i,e} \Py _{i-1}^{(1)}(z_{i-1}) & 2\gamma_{i,e,2}\tilde{v}^{(1)}_{i,e}z^{(1)}_{i-1}v^{(1)}_{i,e}\Py _{i-1}^{(1)}(z_{i-1}) \end{matrix} \right. \\
		 \left. \begin{matrix}\gamma_{i,e,2}\tilde{v}^{(0)}_{i,e}z^{(0)}_{i-1} v^{(2)}_{i,e}\Py _{i-1}^{(2)}(z_{i-1}) & \dots & \gamma_{i,e,p_i}\tilde{v}^{(0)}_{i,e} z^{(0)}_{i-1}v^{(p_i)}_{i,e}\Py _{i-1}^{(p_i)}(z_{i-1}) \end{matrix} \right],
\end{multline}
where $v_{i,e}$ is the $e\textsuperscript{th}$ row of $V^\top_i$, and $\tilde{v}_{i,e}=v_{i,e}\Cp_{i-1}$.
\end{lemma}

\begin{proof}
	The proof is given in Appendix \ref{apx:PITI_f_to_PITI}.
\end{proof}
 Not that if the PITI block $\Sigma_{i-1}^\mathrm{PITI}$  is bilinear then the Koopman embedding $\Sigma_{i}^\mathrm{PITI}$ resulting from \Cref{lemma:rule_PITI_LD} is generally not guaranteed to be bilinear. However, a zero feedtrough term allows to preserve bilinearity:
\begin{corollary}\label{col:piti_f_bilinear}  In case $\Sigma_{i-1}^\mathrm{PITI}$  is bilinear and has no feedtrough, then $\Sigma_{i}^\mathrm{PITI}$ is guaranteed to be bilinear without a feedtrough term, i.e., $\Py _{i}(z_{i})\Ry _{i}(\bar{u}_i)$ is zero.
\end{corollary}
\begin{proof}
	The proof is given in Appendix \ref{apx:PITI_f_to_BLTI}.
\end{proof}

\subsection{Embedding PITI followed by IJ into PITIs} \label{sec:inputjunc}
As a next step, we define the inclusion of an input junction with $m$ branches to a PITI block by repeating the same PITI block $m$-times according to \Cref{fig:PITI_IJ}.
\begin{figure}[ht!]
\begin{center} 
\includegraphics[width=.95\textwidth]{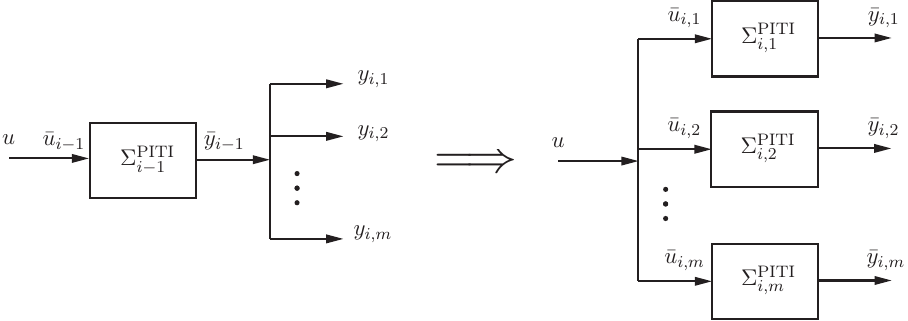}
\caption{Embedding a PITI block followed by a junction into a parallel connection of PITI Koopman representations.}
\label{fig:PITI_IJ}
\end{center}
\end{figure}
\FloatBarrier
\begin{lemma}\label{lemma:rule_piti_ij}
	A PITI block $\Sigma_{i-1}^\mathrm{PITI}$ followed by a junction with $m$ branches according to \eqref{eq:operation_3} can be represented by a set of exact finite-dimensional PITI Koopman representations $\{\Sigma_{i,j}^\mathrm{PITI}\}_{j=1}^m$, each in the form of \eqref{eq:PITI_representation} with states $\{z_{i,j}\equiv z_i\}_{j=1}^m$, 
    inputs $\{\bar{u}_{i,j}\equiv \bar{u}_{i-1}\}_{j=1}^m$, and outputs $\{y_{i,j}\equiv \bar{y}_{i,j} \equiv \bar{y}_{i-1}\}_{j=1}^m$ and each having exactly the same $\Ap_{i,j},\Cp_{i,j}, \Px_{i,j}, \Rx_{i,j},\Py_{i,j}, \Ry_{i,j}.$	
\end{lemma}
\begin{proof}
	It is trivial to see that translating the PITI block after the junction and copying it on each branch maintains the system dynamics, where the outputs are $\{\bar{y}_{i,j}\equiv \bar{y}_{i-1}\}_{j=1}^m$.
\end{proof}
\begin{corollary}\label{corr:rule_piti_ij}
If  $\Sigma_{i-1}^\mathrm{PITI}$ is bilinear, then the Koopman embeddings $\{\Sigma_{i,j}^\mathrm{PITI}\}_{j=1}^m$ resulting from \Cref{lemma:rule_piti_ij} are bilinear and each can be written in the form of \eqref{eq:BLTI_representation}. If $\Sigma_{i-1}^\mathrm{PITI}$ has no feedthrough term, then all $\{\Sigma_{i,j}^\mathrm{PITI}\}_{j=1}^m$ have no feedthrough term, that is, $\Py _{i,j}(z_{i,j})\Ry _{i,j}(\bar{u}_{i,j})=0$.
\end{corollary}
\begin{proof}
	The proof is trivial and follows the same reasoning as Lemma \ref{lemma:rule_piti_ij}.
\end{proof}
Note that if the block chain starts with an input junction, then $\{\bar{u}_{1,j}\equiv u\}_{j=1}^m$ and each branch is initialized according to \Cref{lemma:LTI_OP1} if the next element in the branch is LD or \Cref{lemma:SN_OP2} if the next element in the branch is SN. In case of another input junction in one of the branches, the same operation is repeated.

\subsection{Embedding PITIs followed by OJ into PITI}\label{sec:embedding:OJ}  Next, we define the inclusion of an output junction with $m$ branches, each with a PITI block, into a single PITI block, according to figure \Cref{fig:PITI_OJ}.
\begin{figure}[ht!]
\begin{center}   
\includegraphics[width=.95\textwidth]{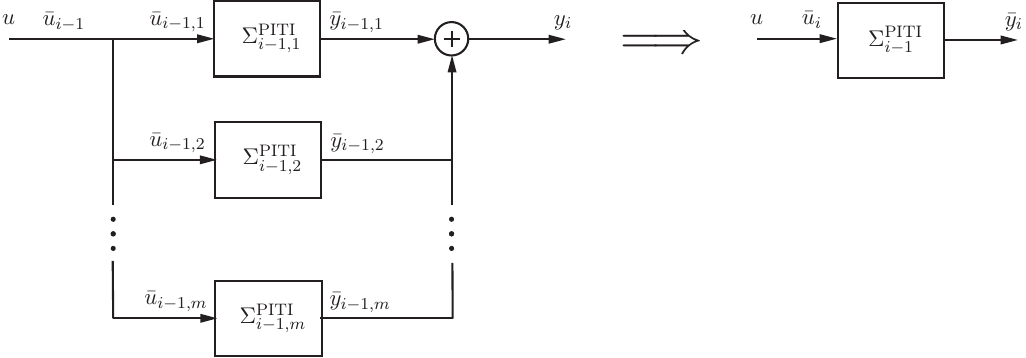}
\caption{Embedding of $m$ paralel branches of PITI blocks via an output junction into a single PITI Koopman representation.}
\label{fig:PITI_OJ}
\end{center}
\end{figure}
\begin{lemma} \label{lemma:rule_PITI_OJ}
The bundle of $m$ parallel branches of PITI blocks $\{\Sigma_{i-1,j}^\mathrm{PITI}\}_{j=1}^m$, sharing the same input, i.e., $\{\bar{u}_{i-1} \equiv \bar{u}_{i-1,j}\}_{j=1}^{m}$, followed by an output junction joining the branches, can be represented by an exact finite dimensional PITI Koopman representation $\Sigma_{i}^\mathrm{PITI}$ in  the form of \eqref{eq:PITI_representation} with state $z_{i} = [\begin{array}{ccc}z_{i-1,1}^\top & \cdots & z_{i-1,m}^\top \end{array} ]^\top $, 
    input $\bar{u}_{i}\equiv \bar{u}_{i-1}$, and output $\bar{y}_{i}\equiv \sum_{j=1}^{m} \bar{y}_{i-1,j}$ and with	
	\begin{subequations}
		\begin{align}
			\Ap  _{i} &=\begin{bmatrix}			\Ap_{i-1,1} & & \\ & \ddots &\\& & \Ap_{i-1,m} 
		\end{bmatrix}, & \Cp_{i} &= \begin{bmatrix}
			\Cp_{i-1,1} &  \cdots & \Cp_{i-1,m}
		\end{bmatrix}, \\
        \Px_i (z_{i}) &\equiv \begin{bmatrix}			\Px _{i-1,1}(z_{i-1,1}) & &\\ & \ddots &\\ & & \Px _{i-1,m}(z_{i-1,m})
		\end{bmatrix}, & \Rx_i (\bar{u}_i) &= \begin{bmatrix}
			\Rx _{i-1,1} (\bar{u}_{i-1}) \\ \vdots \\ \Rx_{i-1,m} (\bar{u}_{i-1})
		\end{bmatrix},\label{eq:OJ:PITI:2} \\
		  \Py _i (z_{i}) &\equiv \begin{bmatrix}
			\Py_{i-1,1}(z_{i-1,1}) & \cdots & \Py _{i-1,m}(z_{i-1,m})
		\end{bmatrix}, & \Ry _i (\bar{u}_i) & \equiv \begin{bmatrix}
			\Ry _{i-1,1} (\bar{u}_{i-1}) \\ \vdots \\ \Ry _{i-1,m} (\bar{u}_{i-1})
		\end{bmatrix}. \label{eq:OJ:PITI:3}
		\end{align}
	\end{subequations}
\end{lemma}
\begin{proof} The resulting representation directly follows from the joint state vector, descried as  $z_{i} = [\begin{array}{ccc}z_{i-1,1}^\top & \cdots & z_{i-1,m}^\top \end{array} ]^\top $, and stacking the state transfers for each  $\{\Sigma_{i-1,j}^\mathrm{PITI}\}_{j=1}^m$ diagonally in the joint state transfer, while the output equation corresponds to stacking the output terms of $\{\Sigma_{i-1,j}^\mathrm{PITI}\}_{j=1}^m$ next to each other, column-wise, corresponding to $\bar{y}_{i}\equiv \sum_{j=1}^{m}\bar{y}_{i-1,j}$.
\end{proof}
\begin{corollary}\label{corr:OJ_OP7}
	If each of the PITI blocks $\{\Sigma_{i-1,j}^\mathrm{PITI}\}_{j=1}^m$ is bilinear, then the Koopman embedding $\Sigma_{i}^\mathrm{PITI}$ resulting from \Cref{lemma:rule_PITI_OJ} is bilinear and  can be written in the form of \eqref{eq:BLTI_representation}. If none of  $\{\Sigma_{i-1,j}^\mathrm{PITI}\}_{j=1}^m$ has feedtrough term, then $\Sigma_{i}^\mathrm{PITI}$ has no feedtrough term, that is, $\Py _{i}(z_{i})\Ry _{i}(\bar{u}_{i})$ is zero.
\end{corollary}
\begin{proof} If all $\{\Sigma_{i-1,j}^\mathrm{PITI}\}_{j=1}^m$ are bilinear, then $\{\Rx _{i-1,1} (u_{i-1})\}_{i=1}^m$ and $\{\Ry _{i-1,1} (u_{i-1})\}_{i=1}^m$ are stacks of identity matrices, making $\Rx _{i}$  and $\Rx _{i-1}$ composed only from constant identity matrices according to \eqref{eq:OJ:PITI:2} and \eqref{eq:OJ:PITI:3}, implying bilinearity of  $\Sigma_{i}^\mathrm{PITI}$. Similarly, if each $\{\Sigma_{i-1,j}^\mathrm{PITI}\}_{j=1}^m$ has no feedtrough term, meaning that all $\{\Py _{i-1,j}(z_{i-1,j})\Ry _{i-1,j}(\bar{u}_{i-1,j})\}_{j=1}^m$ is zero, then due to \eqref{eq:OJ:PITI:3}, $\Py _{i}(z_{i})\Ry _{i}(\bar{u}_{i})$ will correspond to the sum of these zero terms, making it trivially to be zero as well. 
\end{proof}

\subsection{Proving the main results}\label{sec:proof_section}
Now we have all ingredients ready to prove \Cref{thm:PITI_theorem} and \Cref{thm:bilinear_theorem}.
\subsubsection{Proof of \Cref{thm:PITI_theorem}}\label{sec:proof_TH}
	Note that the nonlinear system \eqref{eq:nonlinear_sys_with_input} to be embedded is represented as a block-chain of $\Sigma^{\mathrm{LTI}}$, see \eqref{eq:LTI_block}, and $\Sigma^{\mathrm{NL}}$ blocks, see \eqref{eq:nl_block_def}, in terms of Operations~\crefrange{eq:operation_1}{eq:operation_4}. To prove the statement, we will start from the left of the block-chain with $i=1$ and apply the elementary embeddings~\Crefrange{sec:embedding:LD}{sec:embedding:OJ} corresponding to Operations~\crefrange{eq:operation_1}{eq:operation_4}, until we reach the end of the block chain, i.e., $i=N$. 

    {\bf Step $i=0$:} The start node is just a signal renaming $u_0=u$, corresponding to a technical step. If $N=0$, this concludes the proof as an end node follows the starting node directly, corresponding to $y=u$, which gives a trivial PITI Koopman model with only $\Py\equiv I_{n_\mathrm{u}}$ and $\Ry\equiv I_{n_\mathrm{u}}$, while the rest of the components, including the state dimension, are zero.

    {\bf Step $i=1$:} The start of the block chain can be an LD block $\Sigma_1^\mathrm{LTI}$, an SN block $\Sigma_1^\mathrm{NL}$, or, if $N>1$, an IJ \eqref{eq:operation_3} with $m_1$ branches. 
    In case of an LD block $\Sigma_1^\mathrm{LTI}$, using \Cref{lemma:LTI_OP1}, or in case of an SN block $\Sigma_1^\mathrm{NL}$, using \Cref{lemma:SN_OP2}, the first block element can be embedded in a PITI Koopman representation $\Sigma_1^\mathrm{PITI}$. For $N=1$, this concludes the proof as an end node follows the last block, giving $\Sigma_1^\mathrm{PITI}$ with $u=\bar{u}_1$ and $y=\bar{y}_1$ as the Koopman embedding of the NL system. In case $N>1$, and with the start of an IJ, according to the discussion in \Cref{sec:inputjunc}, each branch can be seen as the start of an individual block chain, for which each element can be embedded according to the above given steps, also applying the IJ rule again if needed. Note that for each branch, the input $\bar{u}_{1,j}$ is equal to $u$, where $j\in\mathbb{I}_1^m$. 

    {\bf Step $i=2$:} The previous part of the block chain, embedded into $\Sigma_{1}^\mathrm{PITI}$ in the previous step, can represent a single PITI  Koopman model or a PITI Koopman model on one of the concurrent branches. $\Sigma_{1}^\mathrm{PITI}$ can be followed by an LD block $\Sigma_2^\mathrm{LTI}$, an SN block $\Sigma_2^\mathrm{NL}$, an OJ \eqref{eq:operation_4} or, for $N>2$, an IJ \eqref{eq:operation_3} with $m_2$ branches. In case of an LD block, the serial connection of $\Sigma_2^\mathrm{LTI}$ and $\Sigma_{1}^\mathrm{PITI}$ is embedded into a PITI Koopman form $\Sigma_{2}^\mathrm{PITI}$ via \Cref{lemma:rule_PITI_LD} with $\bar{u}_2=\bar{u}_1=u$, while, in case of an SN block $\Sigma_2^\mathrm{NL}$, the embedding is accomplished via \Cref{lemma:rule_PITI_SN}. In case of an OJ, the previous branches described by $\Sigma_{1,j}^\mathrm{PITI}$ are jointly represented by $\Sigma_{2}^\mathrm{PITI}$ according to \Cref{lemma:rule_PITI_OJ}. For $N=2$, this concludes the proof as an end node follows the last block, giving $\Sigma_2^\mathrm{PITI}$ with $u=\bar{u}_2$ and $y=\bar{y}_2$ as the Koopman embedding of the NL system. Note that according to the block chain representation, all IJ branches are closed with an OJ before an end node. In case $N>2$ and an IJ, \Cref{lemma:rule_piti_ij} is applied, resulting in  PITI Koopman models $\{\Sigma_{2,j}^\mathrm{PITI}\}_{j=1}^{m_2}$, each with $\bar{u}_{2,j}=u$, and the embedding is continued on the individual branches until an OJ. 

    {\bf Step $i>2$:} If previous parts of the block chain have been embedded into $\Sigma_{i-1}^\mathrm{PITI}$, which can represent a single PITI Koopman model or a model on one of the concurrent branches, a subsequent LD block $\Sigma_{i}^\mathrm{LTI}$ or an SN block $\Sigma_i^\mathrm{NL}$ can be embedded into a PITI Koopman form $\Sigma_{i}^\mathrm{PITI}$ via \Cref{lemma:rule_PITI_LD} or \Cref{lemma:rule_PITI_SN}, respectively. In case of an OJ, the previous branches described by $\Sigma_{i-1,j}^\mathrm{PITI}$ are jointly represented by $\Sigma_{i}^\mathrm{PITI}$ according to \Cref{lemma:rule_PITI_OJ}. For $N=i$, this concludes the proof as an end node follows the last block, giving $\Sigma_i^\mathrm{PITI}$ with $u=\bar{u}_i$ and $y=\bar{y}_i$ as the Koopman embedding of the NL system. In case $N>i$ and an IJ, \Cref{lemma:rule_piti_ij} is applied, resulting in  PITI Koopman models $\{\Sigma_{i,j}^\mathrm{PITI}\}_{j=1}^{m_2}$, each with $u=\bar{u}_{i,j}$, and the embedding is continued on the individual branches until an OJ. 

    This concludes the proof by induction, implying that $\Sigma_N^\mathrm{PITI}$ with $u=\bar{u}_N$ and $y=\bar{y}_N$ is a PITI Koopman embedding of the NL system.
    
\subsubsection{Proof of \Cref{thm:bilinear_theorem}}

Using \Crefrange{corr:LTI_OP1}{corr:OJ_OP7} in combination of \Crefrange{lemma:LTI_OP1}{lemma:rule_PITI_OJ}, it follows by the induction based proof in \Cref{sec:proof_TH} that the resulting $\Sigma_N^\mathrm{PITI}$
is BLTI if each embedding step results in a BLTI model. According to \Crefrange{corr:LTI_OP1}{corr:OJ_OP7}, the BLTI property can be only violated at the SN components, either by starting with an SN component giving a $\Ry(\bar{u}_1)=\tilde{f}_1(u_1)$ that is a polynomial function of $u$, see \Cref{lemma:SN_OP2}, or if a BLTI $\Sigma_{i-1}^\mathrm{PITI}$ is followed by an SN block $\Sigma_{i}^\mathrm{NL}$ and $\Sigma_{i-1}^\mathrm{PITI}$ has a direct feedthrough term, see \Cref{col:piti_f_bilinear}. Now according to \Crefrange{corr:LTI_OP1}{corr:OJ_OP7}, $\Sigma_{i-1}^\mathrm{PITI}$ is guaranteed to be bilinear without a feedthrough term if all previous LD block chain elements had no feedthrough term and the starting block is not an SN which would introduce a non-blinear feedthrough. This concludes the proof.  

\section{Examples} \label{sec:examples}
In this section, we give two examples to illustrate the Koopman embedding by the proposed method. First, we show in detail how a classical MIMO Wiener-Hammerstein system is processed by the iterative Koopman embedding approach. Then, we show the embedding of a complex interconnection of SISO blocks without feedthrough, giving a bilinear Koopman model.

\setcounter{subsection}{1}

\subsubsection{Koopman embedding of a MIMO Wiener-Hammerstein system}\label{sec:MIMO_example-derivation_steps}
We consider the embedding of a Wiener-Hammerstein system which is 
the series interconnection of an LTI block $\Sigma^\mathrm{LTI}_1$, a static nonlinearity $\Sigma^\mathrm{NL}_2$, and an LTI block $\Sigma^\mathrm{LTI}_3$, as can be seen in \Cref{fig:MIMO_example}. We show that such an interconnection can be exactly described as a Koopman PITI model \eqref{eq:PITI_representation} if $\Sigma^\mathrm{NL}_2$ is polynomial. Furthermore, if the linear blocks do not have feedthrough, the embedding becomes bilinear. \par
\begin{figure}[ht!]
\begin{center}
\includegraphics[width=.75\textwidth]{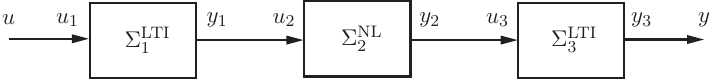}
\caption{Block-chain form of the MIMO Wiener-Hammerstein system.} 
\label{fig:MIMO_example}
\end{center}
\end{figure}
\begin{figure}[htbp]
    \centering
    
    \begin{subfigure}{\textwidth}
        \centering
        \includegraphics[width=0.75\textwidth]{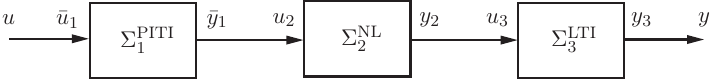}
        \vspace{.3em}
        \caption{Step 1: Embedding $\Sigma^{\text{LTI}}_1$ into $\Sigma^{\text{PITI}}_1$.}
        \label{fig:MIMO_example_step1}
    \end{subfigure}
    
    \vspace{1em} 
    
    \begin{subfigure}[t]{0.58\textwidth}
        \centering
        \includegraphics[width=0.9\textwidth]{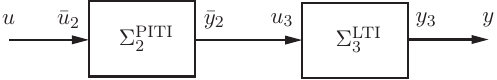}
        \vspace{.3em}
        \caption{Step 2: Embedding $\Sigma^{\text{PITI}}_1$ followed by $\Sigma^\text{NL}_2$ into $\Sigma^{\text{PITI}}_2$.}
        \label{fig:MIMO_example_step2}
    \end{subfigure}
    \hfill
    \begin{subfigure}[t]{0.38\textwidth}
        \centering
        \includegraphics[width=0.8\textwidth]{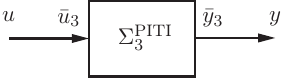}
        \vspace{.3em}
        \caption{Step 3: Embedding $\Sigma^{\text{PITI}}_2$ followed by $\Sigma^\text{LTI}_3$ into $\Sigma^{\text{PITI}}_3$.}
        \label{fig:MIMO_example_step3}
    \end{subfigure}
    
    \caption{Embedding steps of the Wiener-Hammerstein system into a PITI Koopman model.}
    \label{fig:MIMO_example_steps}
\end{figure}
The two LTI blocks $\Sigma^\text{LTI}_1$ and $\Sigma^\text{LTI}_3$ are considered to be: 
\begin{equation}  \label{eq:exap1:LTI}
y_{1} = \underbrace{\left[\begin{array}{c|c} \Al_1 & \Bl_1  \\ \hline \Cl_1 & \Dl_1 \end{array}  \right]}_{G_1} u_{1}, \qquad \text{and} \qquad y_{3} = \underbrace{\left[\begin{array}{c|c} \Al_3 & \Bl_3  \\ \hline \Cl_3 & \Dl_3 \end{array}  \right]}_{G_3} u_{3}, 
\end{equation}
with
\begin{equation*}
    \begin{split} 
\Al_1&=\begin{bmatrix}
	 -0.5& -0.9\\ 2& -0.3
\end{bmatrix},\quad  \Bl_1 = \begin{bmatrix}
	1.2 & -1.5 \\ 0.3 & 1.1
\end{bmatrix},\quad \;\;\,\Cl_1=\begin{bmatrix}
	1 & 0\\0 & 1
\end{bmatrix},\quad \Dl_1=\begin{bmatrix}-0.1 & 0.5 \\ 0.3 & -0.4\end{bmatrix},\\
\Al_3&=\begin{bmatrix}
	 -0.2& -2\\ 0&-0.7
\end{bmatrix},\quad  \Bl_3 = \begin{bmatrix}
	-1.5 & 0.7 \\ 1.4 & -0.3
\end{bmatrix},\quad \Cl_3=\begin{bmatrix}
	1 & 0\\0 & 1
\end{bmatrix},\quad \Dl_3=\begin{bmatrix}0.1 & 0.2 \\ -0.3 & 0.2\end{bmatrix}.
    \end{split} 
\end{equation*}
and $u_t,u_{1,t},u_{3,t}\in\mathbb{R}^2$, $y_t,y_{1,t},y_{3,t}\in\mathbb{R}^2$, and $x_{1,t},x_{3,t}\in\mathbb{R}^2$. We consider the NL block $\Sigma^\text{NL}_2$ to be defined as in Example 1 in \cite{Dreesen}\footnote{Some typos in the coefficients are corrected w.r.t.~the original example in \cite{Dreesen}.}:  
\begin{multline}
    y_2 = f_2(u_2)= \\
    \begin{bmatrix}
	-108u_{2,1}^3 \!-\! 108u_{2,1}^2 u_{2,2} \!+\! 8u_{2,1}^2 \!-\! 36u_{2,1} u_{2,2}^2 \!+\! 16u_{2,1} u_{2,2} \!+\! 12u_{2,1} \!-\! 4u_{2,2}^3 \!+\! 8u_{2,2}^2 \!+\! 8u_{2,2} \!+\! 1 \\ 
    54u_{2,1}^3 \!+\! 54u_{2,1}^2 u_{2,2} \!-\! 24u_{2,1}^2  \!+\! 18u_{2,1} u_{2,2}^2 \!-\! 48u_{2,1} u_{2,2} \!-\! 21u_{2,1} \!+\! 2u_{2,2}^3 \!-\! 24u_{2,2}^2 \!-\! 19u_{2,2} \!-\! 3
\end{bmatrix}
\end{multline}
where $u_{2,1}$ and $u_{2,2}$ are the elements of $u_2$, and $f_2:\mathbb{R}^2\rightarrow\mathbb{R}^2$. The decomposition of $f_2$ is given by:
\begin{equation}\label{eq:decomposition_f_in_g_MIMO_example}
    \underbrace{\begin{bmatrix}
        y_{2,1} \\ y_{2,2} 
    \end{bmatrix}}_{y_2}= \underbrace{\begin{bmatrix}
             1 & 2 \\ -3 & -1
        \end{bmatrix}}_{W_2}\underbrace{\begin{bmatrix}
            2\sigma_{2,1}^2 - 3 \sigma_{2,1} + 1 \\ 2\sigma_{2,2}^3 - \sigma_{2,2}
        \end{bmatrix}}_{g_2(\sigma_2)}, \quad \text{with} \quad \underbrace{\begin{bmatrix}
            \sigma_{2,1} \\ \sigma_{2,2}
        \end{bmatrix}}_{\sigma_2}=\underbrace{\begin{bmatrix}
            -2 & -2\\ -3 & -1
        \end{bmatrix}}_{V^\top_2}\underbrace{\begin{bmatrix}
            u_{2,1} \\ u_{2,2}
        \end{bmatrix}}_{u_2}
\end{equation}
and the coefficients of $g_2$ are described in \Cref{tbl:table_example_MIMO}. Note that, for the given $f_2$ and associated decomposition via $g_2$, the total degree is $p_2=3$.

\par
\begin{table}[ht!]\caption{Coefficients of $g_2$ in the monomial decomposition in the Wiener-Hammerstein system example.}\label{tbl:table_example_MIMO}
\begin{center}
\begin{tabular}{ |c|c|c|c| } 
 \hline
 $\gamma_{2,1,0}=1$ & $\gamma_{2,1,1}=-3$ & $\gamma_{2,1,2}=2$ & $\gamma_{2,1,3}=0$\\ \hline
 $\gamma_{2,2,0}=0$ & $\gamma_{2,2,1}=-1$ & $\gamma_{2,2,2}=0$ & $\gamma_{2,2,3}=2$ \\ 
 \hline
\end{tabular}
\end{center}
\end{table}

By the block-chain structure in \Cref{fig:MIMO_example}, the first step in computing a Koopman embedding of the system is to 
convert $\Sigma^\text{LTI}_1$ to $\Sigma^\text{PITI}_1$ according to \Cref{sec:embedding:LD}. Following \Cref{lemma:LTI_OP1}, the dynamics of $\Sigma^{\text{PITI}}_1$ are described by:
\begin{subequations}
\begin{align}
        \dot{z}_1&=\Ap_1 z_1 + \Px_1(z_1)\Rx_1(\bar{u}_1)\bar{u}_1,\\
        \bar{y}_1 &=\Cp_1 z_1 + \Py_1(z_1)\Ry_1(\bar{u}_1)\bar{u}_1,
    \end{align}
    \end{subequations}
with $z_1=x_1$, $\bar{y}_1=y_1$, and $\bar{u}_1 = u_1 = u$. For the state equation, $\Ap_1 = \Al_1$, $\Px_1(z_1) = \Bl_1$, $\Rx_1(\bar{u}_1)=I_{n_{\bar{\mr u},1}}$, while for the output equation, $\Cp_1=\Cl_1$, $\Py_1(z_1)=\Dl_1$, and $\Ry_1(\bar{u}_1)=I_{n_{\bar{\mr u},1}}$. This results in the block-chain in \Cref{fig:MIMO_example_step1}, accomplishing Step 1 of the embedding process.

\par
As $\Sigma^\text{PITI}_1$ is followed by $\Sigma^\text{NL}_2$ in \Cref{fig:MIMO_example_step1}, we embed these two blocks
into $\Sigma^\text{PITI}_2$ according to \Cref{sec:embedding:PITI:SN}. The interconnection between $\Sigma^{\text{PITI}}_1$ and $\Sigma^{\text{NL}}_2$ is described by the equations:
\begin{subequations}
\begin{align}
                \dot{z}_1&=\Ap_1 z_1 + \Px_1(z_1)\Rx_1(\bar{u}_1)\bar{u}_1,\\
        \bar{y}_1 &=\Cp_1 z_1 + \Py_1(z_1)\Ry_1(\bar{u}_1)\bar{u}_1,  \\
        y_2 &= f_2(\bar{y}_1)=W_2g_2(V^\top_2\bar{y}_1),  
        = W_2 \begin{bmatrix}
            g_{2,1}(\sigma_{2,1}) & g_{2,2}(\sigma_2,2)
        \end{bmatrix},
    \end{align}
    \end{subequations}
where $\sigma_{2,e}=v_{2,e}\bar{y}_1$ and $v_{2,e}$ is the $e\textsuperscript{th}$ row of $V^\top_2$, with $e\in\{1,2\}$. We can write each $g_{2,e}$ as
\begin{multline}\label{eq:g_2_e_decomp_MIMO_example}
        g_{2,e}(\sigma_{2,e})\stackrel{\mathmakebox[\widthof{=}]{\eqref{eq:decomp_ge_piti_nl}}}{=}\; \gamma_{2,e,0} + \gamma_{2,e,1}\sigma_{2,e} + \gamma_{2,e,2}\sigma^2_{2,e} + \gamma_{2,e,3}\sigma^3_{2,e} =\\
        \gamma_{2,e,0} + \gamma_{2,e,1}\left(\tilde{v}_{2,e}z_1 + \hat{v}_{2,e}\bar{u}_1\right) + \gamma_{2,e,2}\left(\tilde{v}_{2,e}z_1 + \hat{v}_{2,e}\bar{u}_1\right)^2 + \gamma_{2,e,3}\left(\tilde{v}_{2,e}z_1 + \hat{v}_{2,e}\bar{u}_1\right)^3
    \end{multline}
with $\tilde{v}_{2,e} =  v_{2,e}\Cp_1= v_{2,e}\Cl_1$ and $\hat{v}_{2,e}=v_{2,e}\Py_1(z_1)=v_{2,e}\Dl_1$.
Following \Cref{lemma:rule_PITI_SN},
\begin{equation*}
z_2 = \begin{bmatrix}
    1 & z^\top_1 & \left(z^{(2)}_1\right)^\top & \left(z^{(3)}_1\right)^\top
\end{bmatrix}^\top,
\end{equation*}
is the new state and the output equation of $\Sigma^{\text{PITI}}_2$ is defined by the functions:
\begin{equation}\label{eq:piti_2_output_eq_MIMO_example_eq1}
    \begin{split}
        \Cp_2 &= W_2\Gamma_2=W_2\begin{bmatrix}
            \gamma_{2,1,0} & \gamma_{2,1,1}\tilde{v}_{2,1} & \gamma_{2,1,2}\tilde{v}_{2,1}^{(2)} & \gamma_{2,1,3}\tilde{v}_{2,1}^{(3)} \\
            \gamma_{2,2,0} & \gamma_{2,2,1}\tilde{v}_{2,2} & \gamma_{2,2,2}\tilde{v}_{2,2}^{(2)} & \gamma_{2,2,3}\tilde{v}_{2,2}^{(3)}
        \end{bmatrix}, \\
        \Py_2(z_2)&=W_2\begin{bmatrix}
            \Py_{2,1}(z_1) \\ \Py_{2,2}(z_1)
        \end{bmatrix}, \quad \Ry_2(\bar{u}_2)=\begin{bmatrix}
            I_{n_{\bar{\mr u},2}} \\
            I_{n_{\bar{\mr u},2}} \\
            I_{n_{\bar{\mr u},2}}^{(2)} \left(I_{n_{\bar{\mr u},2}} \otimes \bar{u}_2\right)\\
            I_{n_{\bar{\mr u},2}}\\
            I_{n_{\bar{\mr u},2}}^{(2)} \left(I_{n_{\bar{\mr u},2}} \otimes \bar{u}_2\right)\\
            I_{n_{\bar{\mr u},2}}^{(3)} \left(I_{n_{\bar{\mr u},2}} \otimes \bar{u}_2^{(2)}\right)\\
        \end{bmatrix},
    \end{split}
\end{equation}
where we used $\Ry_1(\bar{u}_1)=I_{n_{\bar{\mr u},1}}$, $\bar{u}_2=\bar{u}_1=u$, and, based on \eqref{eq:g_2_e_decomp_MIMO_example} and \eqref{eq:eq_linear_in_z_output_ie}:
\begin{multline}\label{eq:piti_2_output_eq_MIMO_example_eq2}
	\Py _{2,e}(z_{1})=\left[ \begin{matrix}
		\gamma_{2,e,1}\tilde{v}^{(0)}_{2,e}z^{(0)}_{1}\hat{v}_{2,e}^{(1)} & 2\gamma_{2,e,2}\tilde{v}^{(1)}_{2,e}z^{(1)}_{1}\hat{v}_{2,e}^{(1)} & \gamma_{2,e,2}\tilde{v}^{(0)}_{2,e}z^{(0)}_{1}\hat{v}_{2,e}^{(2)}\end{matrix} \right. \\
		 \left. \begin{matrix}3\gamma_{2,e,3}\tilde{v}^{(2)}_{2,e}z^{(2)}_{1}\hat{v}_{2,e}^{(1)} & 3\gamma_{2,e,3}\tilde{v}^{(1)}_{2,e}z^{(1)}_{1}\hat{v}_{2,e}^{(2)} & \gamma_{2,e,3}\tilde{v}^{(0)}_{2,e}z^{(0)}_{1}\hat{v}_{2,e}^{(3)}\end{matrix} \right],
\end{multline}
where $\hat{v}_{2,e}^{(k)} = \left(v_{2,e}\Dl_1\right)^{(k)}= v_{2,e}^{(k)}\Dl^{(k)}_1$. Next, the state equation, according to \Cref{lemma:rule_PITI_SN}, is defined by
\begin{equation}
    \Ap_2 =\! \!\begin{bmatrix}
        0 & & & \\
        & \Ap_1 &  & \\
        & & {}^2\Ap_1 \\
        & & & {}^3\Ap_1
    \end{bmatrix}\!\!=\!\! \begin{bmatrix}
        0 & & & \\
        & \Al_1 &  & \\
        & & {}^2\Al_1 \\
        & & & {}^3\Al_1
    \end{bmatrix}\!, \; \Px_2(z_2)=\!\!\begin{bmatrix}
        0_{1\times n_{\mr z , 1}}\\ I_{n_{\mr z , 1}} \\ \frac{\partial z_1^{(2)}}{\partial z_1} \\ \frac{\partial z_1^{(3)}}{\partial z_1}
    \end{bmatrix}\cdot \Px_1(z_1) = \!\!\begin{bmatrix}
        0_{1\times n_{\mr z , 1}}\\ I_{n_{\mr z , 1}} \\ \frac{\partial z_1^{(2)}}{\partial z_1} \\ \frac{\partial z_1^{(3)}}{\partial z_1}
    \end{bmatrix}\!\!\Bl_1
\end{equation}
and $\Rx_2(\bar{u}_2)=\Rx(\bar{u}_1)=I_{n_{\bar{\mr u},2}}$ with $\bar{u}_2=\bar{u}_1=u$. We can further simplify $\Px_2(z_2)$ as follows. For $j\in\{2,3\}$, considering ${}_k\Bzct_1$ is the $k\textsuperscript{th}$ column of $\Bzct_1=\Bl_1$, based on \Cref{lemma:appendix_lemma_partial_B}, we have that: 
\begin{equation}
        \frac{\partial z_1^{(j)}}{\partial z_1}{}_k\Bzct_1 =\underbrace{\left(\sum^{j-1}_{\tau=0}I_{n_{\mr z ,1}}^{(\tau)}\otimes  {}_k\Bzct_1 \otimes I_{n_{\mr z ,1}}^{(j-\tau-1)}\right)}_{{}^j_k\Bzct_1}z_1^{(j-1)}. 
\end{equation}
Stacking all components and using that $n_{\bar{\mr u},2}=n_{\mr u}=2$, we can write the state equation of $\Sigma^{\text{PITI}}_2$ as:
\begin{equation}
    \dot{z} = \Ap_2 + \sum^{2}_{k=1} {}_k\Bz_2 z_2 \bar{u}_{2,k}, \quad \text{with} \quad {}_k\Bz_2 = \begin{bmatrix}
        0 & &  & \\
        {}_k\Bzct_1 & 0 & & \\
        & {}^2_k\Bzct_1 & 0 & \\
        & & {}^3_k\Bzct_1 & 0
    \end{bmatrix}.
\end{equation}
Equivalently, we can write the state equation as in \eqref{eq:BLTI_representation}:
\begin{equation}
    \dot{z}_2 = \Ap_2z_2 + \sum^{n_{\mr z,2}}_{j=1}\Bz_{2,j}z_{2,j}\bar{u}_2,
\end{equation}
with $z_{2,j}$ being the $j^\textsuperscript{th}$ element of $z_2$ and $\Bz_{2,j}=[{}_{1,j}\Bz_{2} \;\; {}_{2,j}\Bz_{2}]\in\mathbb{R}^{n_{\mr z, 2}\times n_{\mr u}}$, where ${}_{k,j}\Bz_{2}$ is the $j\textsuperscript{th}$ column of ${}_k\Bz_2$, with $k\in\{1,2\}$, as $n_{\bar{\mr u},2}=n_{\bar{\mr u},1}=n_{\mr u}=2$. Overall, the state-space representation of the block $\Sigma^{\text{PITI}}_2$ is 
\begin{subequations}
\begin{align}
        \dot{z}_2 &= \Ap_2z_2 + \Px_2(z_2)\Rx_2(\bar{u}_2)\bar{u}_2
        =\Ap_2z_2 + \left(\sum^{n_{\mr z,2}}_{j=1}\Bz_{2,j}z_{2,j}\right) \cdot I_{n_{\bar{\mr u},2}} \cdot\bar{u}_2 ,\\
        \bar{y}_2 &= \Cp_2 z_2 + \Py_2(z_2)\Ry_2(\bar{u}_2)\bar{u}_2,
\end{align}
    \end{subequations}
with $\Px_2(z_2) = \sum^{n_{\mr z,2}}_{j=1}\Bz_{2,j}z_{2,j}$ and $\Rx_2(\bar{u}_2)=I_{n_{\bar{\mr u},2}}$ showing the linearity of $\Px_2(z_2)$ in $z_2$. This results in the block-chain in \Cref{fig:MIMO_example_step2}, accomplishing Step 2 of the embedding process.

\par

Next, as $\Sigma^\text{PITI}_2$ is followed by $\Sigma^{\text{LTI}}_3$ in \Cref{fig:MIMO_example_step2}, we embed these two blocks
into $\Sigma^\text{PITI}_3$ according to \Cref{sec:embedding:PITI:LD}.
Based on \Cref{lemma:rule_PITI_LD}, the dynamics of $\Sigma^{\text{PITI}}_3$ are given by:
\begin{subequations}
\begin{align}
        \dot{z}_3 &= \Ap_3z_3 + \Px_3(z_3)\Rx_3(\bar{u}_3)\bar{u}_3,\\
        \bar{y}_3&= \Cp_2 z_3 + \Py_3(z_3)\Ry_3(\bar{u}_3)\bar{u}_3,
\end{align}
    \end{subequations}
with $z_3 = [z_2^\top \;\; x_3^\top]^\top$, $\bar{y}_3=y_3$, $\bar{u}_3=\bar{u}_2=u$, and 
\begin{subequations}\label{eq:piti_2_linear_3_example_matrices}
	\begin{align}
	\Ap  _3&= \begin{bmatrix}
		\Ap  _{2} & 0 \\ \Bl_3\Cp_{2} & \Al_3
	\end{bmatrix},& \Px _3(z_3) &\!=\! \begin{bmatrix}
		\Px _{2}(z_{2}) & 0 \\ 0 & \Bl_3\Py _{2}(z_{2})
	\end{bmatrix},& \Rx _3(\bar{u}_3)&\!=\!\begin{bmatrix}
		\Rx _{2}(\bar{u}_{2}) \\ \Ry _{2}(\bar{u}_{2})
	\end{bmatrix}\\
	\Cp_3 &=\begin{bmatrix}
		\Dl_3\Cp_{2} & \Cl_3
	\end{bmatrix},& \Py _3(z_3)&\!=\! \Dl_3\Py _{2}(z_{2}),& \Ry _3(\bar{u}_3)&\!=\!\Ry _{2} (\bar{u}_{2}).
	\end{align}
\end{subequations}

\par

Note that, if the linear blocks $\Sigma^{\text{LTI}}_1$ and $\Sigma^{\text{LTI}}_3$ have no feedthrough, i.e., $\Dl_1=0_{n_{\mr y , 1}\times n_{\mr u , 1}}$ and $\Dl_3 = 0_{n_{\mr y , 3}\times n_{\mr u , 3}}$, then, based on \eqref{eq:piti_2_linear_3_example_matrices}, $\Py_3(z_3)\Ry(\bar{u}_3)=0_{n_{\bar{\mr y},3}\times n_{\bar{\mr u},3}}$. Moreover, based on \eqref{eq:piti_2_output_eq_MIMO_example_eq1} and \eqref{eq:piti_2_output_eq_MIMO_example_eq2}, and the fact that all elements $\hat{v}_{2,e}=v_{2,e}\Dl_1$ are zero, $\Py_2(z_2)\Ry_2(\bar{u}_2)=0_{n_{\bar{\mr y},2}\times n_{\bar{\mr u},2}}$ and the resulting dynamics of $\Sigma^{\text{PITI}}_3$ are bilinear. Hence, using $\Px_2(z_2) = \sum^{n_{\mr z,2}}_{j=1}\Bz_{2,j}z_{2,j}$ and $\Rx_2(\bar{u}_2)=I_{n_{\bar{\mr u},2}}$,  $\Sigma^{\text{PITI}}_3$ can be written in the following bilinear form denoted as $\Sigma^{\text{BLTI}}_3$:
\begin{subequations}
\begin{align}
        \dot{z}_3 &= \Ap_3z_3 + \sum^{n_{\mr z , 3}}_{j=1}\Bz_{3,j}z_{3,j}\bar{u}_3
        =\Ap_3  z_3 + \sum^2_{k=1}{}_k\Bz_{3}z_3\bar{u}_{3,k},\\
        \bar{y}_3 &= \Cp_3 z_3,
\end{align}
    \end{subequations}
where we used $n_{\bar{\mr u},3}=n_{\mr u}=2$ and, based on the embedding of $\Sigma^{\text{PITI}}_2$ and \eqref{eq:piti_2_linear_3_example_matrices}:
\begin{equation*}
    {}_k \bar{B}_3 = \begin{bmatrix}
        {}_k\bar{B}_2 & 0_{n_{\mr z , 2}\times n_{\mr x , 3}} \\ 0_{n_{\mr x , 3}\times n_{\mr z , 2}} & 0_{n_{\mr x , 3} \times n_{\mr x , 3}}
    \end{bmatrix}\qquad \text{and} \qquad \Bz_{3,j} = \begin{bmatrix}
        \Bz_{2,j} \\ 0_{n_{\mr x , 3} \times n_{\bar{\mr u},3}}
    \end{bmatrix}.
\end{equation*}
Moreover, $\bar{B}_{3,j}=[{}_{1,j}\bar{B}_3 \; \; {}_{2,j}\bar{B}_3]$, where ${}_{k,j}\bar{B}_3$ is the $j\textsuperscript{th}$ column of ${}_k\Bz_{3}$, with $k\in\{1,2\}$. This results in the block-chain in \Cref{fig:MIMO_example_step3} with $u=\bar{u}_3$ and $y=\bar{y}_3$, accomplishing the final step of the embedding process as there are no more blocks to embed.

\subsubsection{Final model}\label{sec:final_models_MIMO_example}
\begin{figure}[htbp]
    \centering
    \begin{subfigure}{0.48\textwidth}
        \centering
        \includegraphics[width=\textwidth]{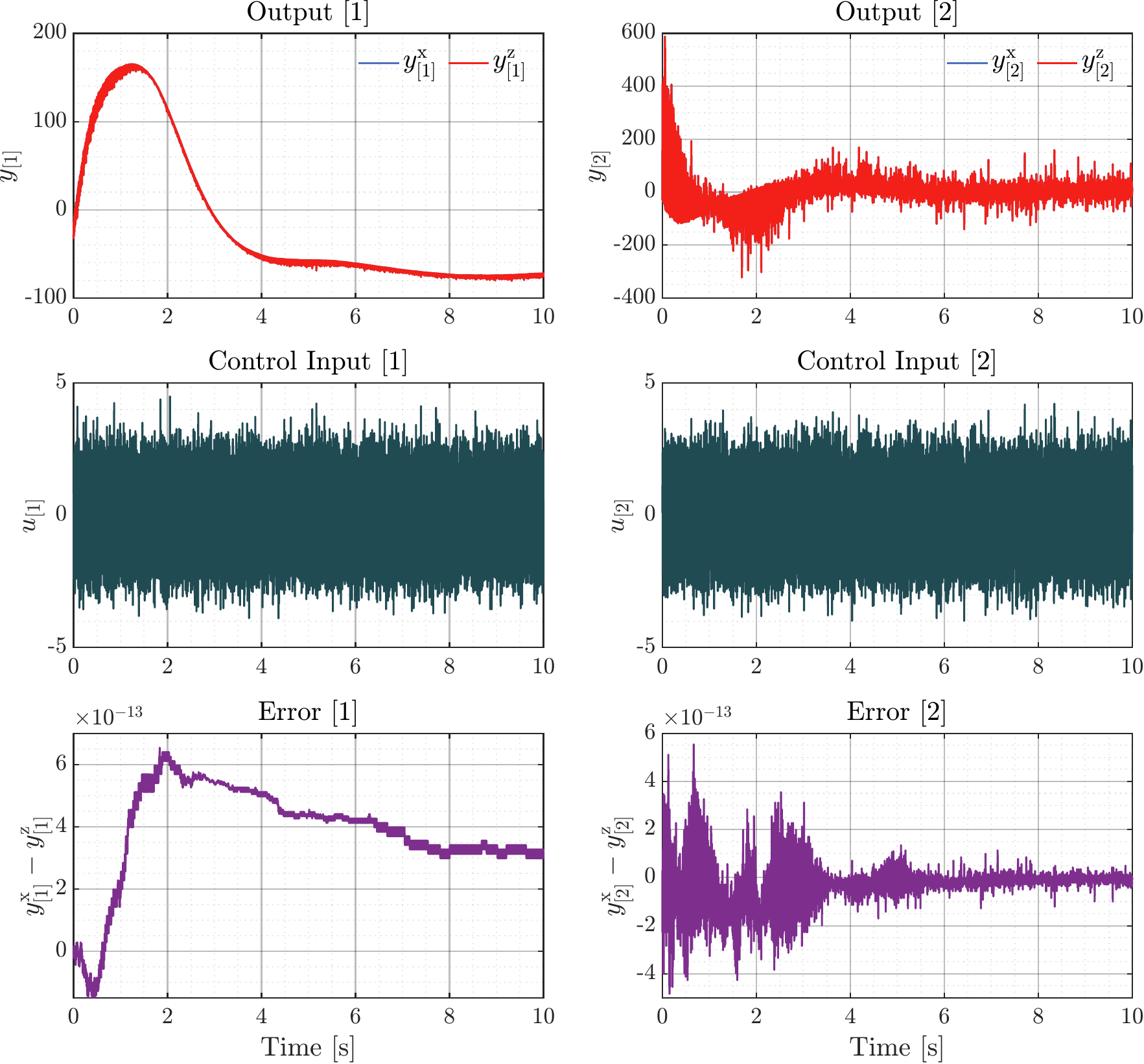}
        \caption{Comparison with the Koopman PITI embedding.\vspace{-.5cm}} 
        \label{fig:MIMO_results_example_figure_subfigure_1}
    \end{subfigure}
    \hfill
    \begin{subfigure}{0.48\textwidth}
        \centering
        \includegraphics[width=\textwidth]{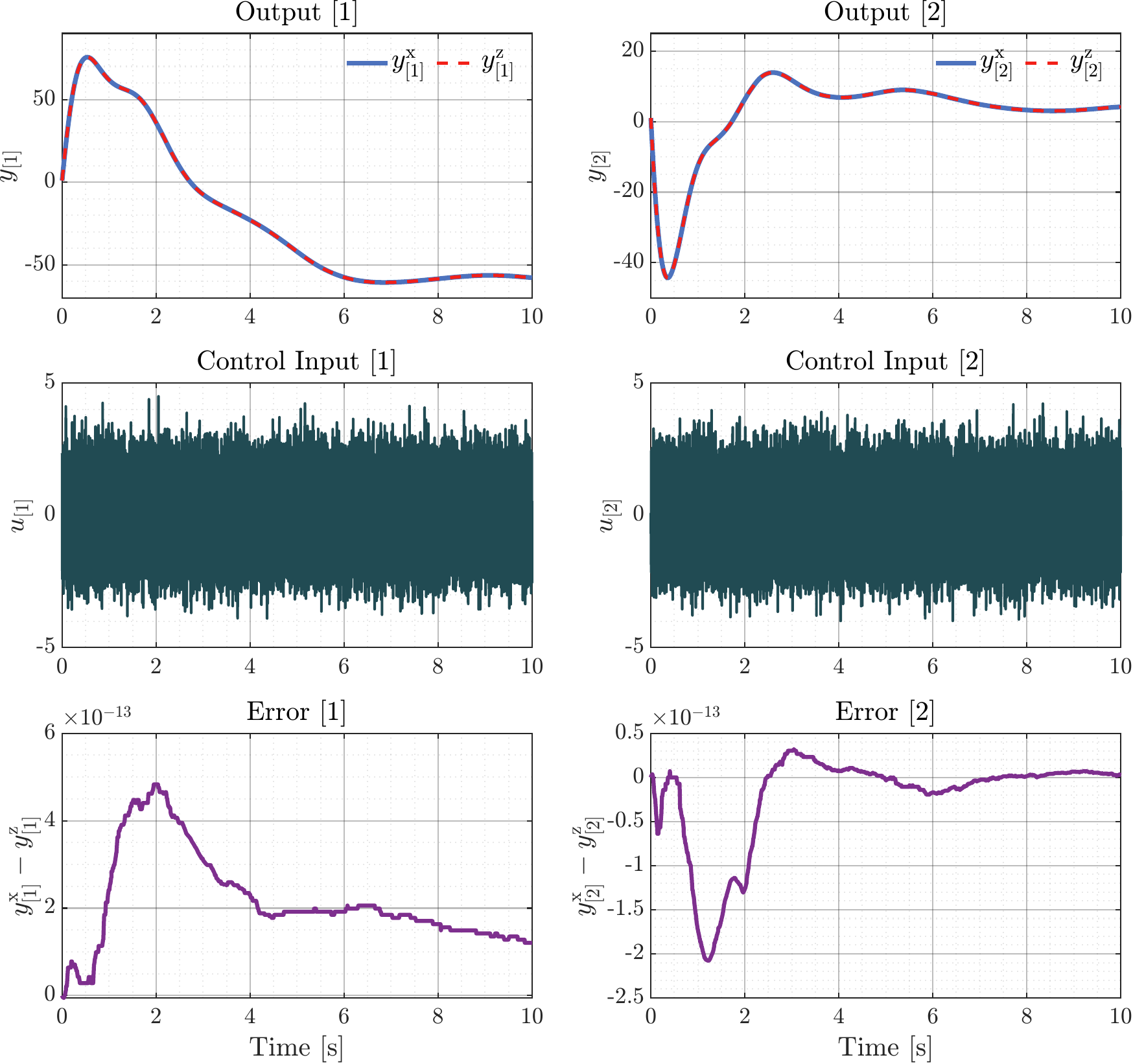}
        \caption{Comparison with the Koopman BLTI embedding. \vspace{-.5cm}} 
        \label{fig:MIMO_results_example_figure_subfigure_2}
    \end{subfigure}
    \caption{Simulated output responses of the MIMO Wiener-Hammerstein block chain system  depicted in \Cref{fig:MIMO_example} and (a) the embedded $\Sigma^{\text{PITI}}$ model and (b) the $\Sigma^{\text{BLTI}}$ model corresponding to the simplified system with $\Dl_1=\Dl_2=0_{2\times 2}$. The responses of nonlinear system ($y^\mathrm{x}$) and the Koopman models ($y^\mathrm{z}$) are given (top plots) for the white Gaussian noise input signal $u$ (middle plot) with the difference of the obtained responses (bottom plot) also depicted. \vspace{-.6cm}}
    \label{fig:MIMO_results_example_figure}
\end{figure}
It is important to note that, due to the nature of the Kronecker product, the resulting lifted state $z_3=\Phi(x)$ contains duplicate states (e.g., for $x_1$ and $x_2$ being scalar elements of a vector $x=[
   \  x_1 \ x_2 \ ]^\top$, $x\otimes x = \{x_1^2,x_1x_2,x_2x_1,x_2^2\}$ contains the term $x_1x_2$ twice). In terms of a post processing step for the resulting $\Sigma^{\text{PITI}}_3$, its is simple to
 remove the duplicate states by constructing an appropriate state projection:  
$z=Tz_3$ with $z_3=T^\dagger z$, where $T\in\mathbb{R}^{n_\mr z \times n_{\mr z,3}}$ is a matrix that selects the unique elements (in each row it contains only zeros except for one element which is one) and $T^\dagger$ is its inverse. This gives the resulting $\Sigma^{\text{PITI}}$ as
\begin{subequations}\label{eq:example_mimo_piti_reduced}
\vspace{-.3cm}
    \begin{align}
        \dot{z}&= \Ap z + \Px(z) \Rx (u)u\\
        y &=  \Cp z +  \Py (z) \Ry(u)u
    \end{align}
\end{subequations}
with $\Ap = T\Ap_3T^\dagger$, $ \Cp = \Cp_3T^\dagger$, $\Px(z) :=T\Px_3(T^\dagger z)$, $ \Rx(u)=\Rx_3(u)$, $ \Py(z):=\Py_3(T^\dagger z)$, and $ \Ry(u)=\Ry_3(u)$. In this example, through this projection, the lifted state is reduced from $n_{\mr z,3}=17$ to $n_{\mr z}=12$. The same projection applies for the simplified BLTI form $\Sigma^{\text{BLTI}}_3$ of the dynamics to get the final $\Sigma^{\text{BLTI}}$ embedding.

\par

To validate the models, we simulate the response of the nonlinear system depicted in \Cref{fig:MIMO_example} to an i.i.d.~input signal $u_{t} \sim \mathcal{N}(0,I_2)$  with Runge Kutta 4 numerical integration using a step size of $\delta t=10^{-4}$s under both \eqref{eq:exap1:LTI} with the original $\Dl_i$ matrices and also with $\Dl_1$ and $\Dl_2$ set to zero. 
We apply the same excitation signal and numerical integration to the obtained $\Sigma^{\text{PITI}}$ model (corresponding to the full system) and the $\Sigma^{\text{BLTI}}$ model (corresponding to $\Dl_1$ and $\Dl_2$ set to zero in the LTI blocks). The used initial conditions are $x_{1,0}=x_{2,0}=[1\;\; 1]^\top$ and $z_0 = Tz_{3,0}$.  
Comparing the responses in \Cref{fig:MIMO_results_example_figure} shows that in both cases, PITI (left) and BLTI (right), the error between the output signals $y^\mr x_{[i]}$ of the original nonlinear system with $i\in\{1,2\}$ denoting the elements, and $y^\mr z_{[i]}$ of the embedded models is around $10^{-13}$ in magnitude, which is close to the machine precision of the involved numerical computations. This indicates that $\Sigma^{\text{PITI}}$ and $\Sigma^{\text{BLTI}}$ are exact embeddings of the original system.

\vspace{-.1cm}

\stepcounter{subsection}          
\setcounter{subsubsection}{0}     

\subsubsection{Koopman embedding of a SISO block chain system}
As a second example consider the nonlinear block chain model 
 with $n_\mathrm{u}=1$ and $n_\mathrm{y}=1$ given in \Cref{fig:SISO_example_1}.

\begin{figure}[ht!]
\begin{center}   
\includegraphics[width=.95\textwidth]{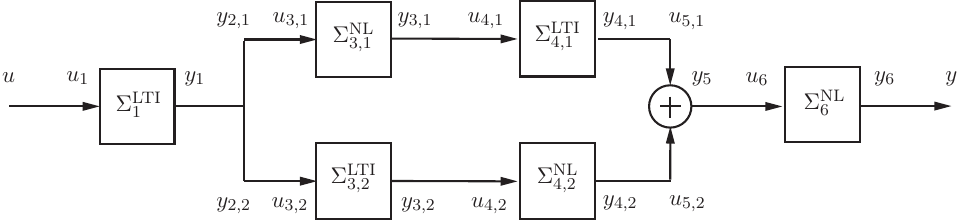}
\caption{Interconnection structure of the SISO nonlinear block chain system.}
\label{fig:SISO_example_1}
\end{center}
\end{figure}
The LTI blocks $\Sigma^\text{LTI}_1$, $\Sigma^\text{LTI}_{3,2}$ and $\Sigma^\text{LTI}_{4,1}$ are defined by the matrices: 
\begin{align*}
\Al_1 &= \begin{bmatrix}
-0.5 & 0 \\ 0 & -0.3
\end{bmatrix} \qquad \;\;\Bl_1 = \begin{bmatrix}
0.2 \\ 0.3
\end{bmatrix} \qquad \;\;\; \;\,\Cl_1 = \begin{bmatrix}
0.4 & 0.6
\end{bmatrix} \qquad \Dl_1 \;\;\,= \begin{bmatrix}
0 \\ 0
\end{bmatrix}\\
\Al_{3,2} &= \begin{bmatrix}
-0.2 & 0 \\ 0 & -0.7
\end{bmatrix} \qquad \Bl_{3,2} = \begin{bmatrix}
-0.5 \\ 0.4
\end{bmatrix} \qquad \Cl_{3,2} = \begin{bmatrix}
0.7 & 0.5
\end{bmatrix} \qquad \Dl_{3,2} = \begin{bmatrix}
0 \\ 0
\end{bmatrix}\\
\Al_{4,1} &= \begin{bmatrix}
-0.4 & 0 \\ 0 & -0.2
\end{bmatrix} \qquad \Bl_{4,1} = \begin{bmatrix}
-1.2 \\ -2
\end{bmatrix} \qquad \Cl_{4,1} = \begin{bmatrix}
1 & 1
\end{bmatrix} \qquad \quad \;\; \Dl_{4,1} = \begin{bmatrix}
0 \\ 0
\end{bmatrix}
\end{align*}
with corresponding input $u_{1,t},u_{3,2,t},u_{4,1,t}\in\mathbb{R}$, output $y_{1,t},y_{3,2,t},y_{4,1,t}\in\mathbb{R}$ and state signals $x_{1,t},x_{3,2,t},x_{4,1,t}\in\mathbb{R}^2$. The NL blocks $\Sigma^\text{NL}_{3,1}$, $\Sigma^\text{NL}_{4,2}$ and $\Sigma^\text{NL}_{6}$ are defined in terms of
\begin{align*}
f_{3,1}(u_{3,1}) &= \gamma_{3,1,0} + \gamma_{3,1,1} u_{3,1} + \gamma_{3,1,2} u_{3,1}^2\\
f_{4,2}(u_{4,2}) &= \gamma_{4,2,0} + \gamma_{4,2,1} u_{4,2} + \gamma_{4,2,2} u_{4,2}^2\\
f_{6}(u_{6}) &= \gamma_{6,0} + \gamma_{6,1} u_{6} + \gamma_{6,2} u_{6}^2
\end{align*}
with signal dimensions $f_{3,1},f_{4,2},f_6:\mathbb{R}\rightarrow \mathbb{R}$ and coefficients that are given in \Cref{tbl:table_coeffs_SISO_example}.
Note that we have only scalar polynomial nonlinearities that do not require decomposition. 
\begin{table}[ht!]\caption{Coefficients of the NL blocks in the monomial decomposition form in the SISO block chain example.}
\label{tbl:table_coeffs_SISO_example}
\begin{center}
\begin{tabular}{ |l|l|l| } 
 \hline
 $\gamma_{3,1,0}=\;\;\,0.2$ & $\gamma_{3,1,1}=-1.2$ & $\gamma_{3,1,2}=\;\;\,0.3$  \\ \hline
 $\gamma_{4,2.0}=-0.3$ & $\gamma_{4,2,1}=\;\;\,0.5$ & $\gamma_{4,2,2}=-0.1$ \\  \hline
 $\;\,\,\gamma_{6,0}=\;\;\,0.5$ & $\;\,\,\gamma_{6,1}=-2.2$ & $\;\,\,\gamma_{6,2}=-0.2$ \\  \hline 
\end{tabular}
\end{center}
\end{table}

\subsubsection{Embedding $\Sigma_1^\text{LTI}$}
According to \Cref{fig:SISO_example_1},
the first step is to embed $\Sigma^{\text{LTI}}_1$ into $\Sigma^{\text{PITI}}_1$ following \Cref{sec:embedding:LD}. However, due to linearity of this LD block and because all LD blocks have no feedtrough, we will accomplish the embedding directly into BLTI blocks $\Sigma^{\text{BLTI}}_i$. Based on 
\Cref{lemma:LTI_OP1,corr:LTI_OP1}, the dynamics of $\Sigma^{\text{BLTI}}_1$ are given by:

\begin{subequations}
\begin{align}
        \dot{z}_1&=\Ap_1 z_1 + \Bzct_1 \bar{u}_1, \\
        \bar{y}_1 &= \Cp z_1,
\end{align}
    \end{subequations}
with $\Ap _1= \Al _1$, $\Bzct_1 = \Bl_1$, $\Cp_1 = \Cl_1$, $z_1=x_1$, $\bar{u}_1=u_1=u$, $\bar{y}_1=y_1$, and  $\Bz_1 = 0_{n_{\mr z , 1} \times n_{\mr z , 1}}$.

\subsubsection{Absorbing the input junction}
For the next step, following  \Cref{lemma:rule_piti_ij} with \Cref{corr:rule_piti_ij}, we obtain two bilinear blocks $\Sigma^{\text{BLTI}}_{2,1} = \Sigma^{\text{BLTI}}_{2,2}=\Sigma^{\text{BLTI}}_{1}$. Thus, we have $z_{2,1}=z_{2,2}=z_1$, $\bar{u}_{2,1}=\bar{u}_{2,2}=\bar{u}_1 = u$, and $\bar{y}_{2,1}=y_{2,1}=\bar{y}_1$, which is equal to $\bar{y}_{2,2}=y_{2,2}=\bar{y}_1$. Furthermore, $\Ap_{2,1} = \Ap_{2,2}=\Ap_1$, $\Bzct_{2,1} = \Bzct_{2,2} = \Bzct_1$, $\Cp_{2,1}=\Cp_{2,2}=\Cp_1$, and $\Bz_{2,1}=\Bz_{2,2}=0_{n_{\mr z , 1} \times n_{\mr z , 1}}$. This results in the block-chain depicted in \Cref{fig:SISO_example_steps1}.

\subsubsection{Branch 1: absorbing $f_{3,1}$}\label{sec:siso_example_LTI1_f1}
Next, according to the systematic conversion process of the block-chain, we absorb $\Sigma^{\text{NL}}_{3,1}$ into the embedding. To do so, we start with collecting the equations defining the interconnection between $\Sigma^{\text{BLTI}}_{2,1}$ and $\Sigma^{\text{NL}}_{3,1}$:
\begin{equation}
	\begin{split}
		\dot{z}_{2,1} &= \Ap_{2,1}z_{2,1} + \Bzct_{2,1} \bar{u}_{2,1}\\
		\bar{y}_{2,1} &= \Cp_{2,1} z_{2,1} \\
		\bar{y}_{3,1} &= f_{3,1}(\bar{y}_{2,1})=\gamma_{3,1,0} + \gamma_{3,1,1}\Cp_{2,1}z_{2,1} + \gamma_{3,1,2}\Cp^{(2)}_{2,1}z^{(2)}_{2,1} 
	\end{split}
\end{equation}
with $\bar{y}_{3,1}=y_{3,1}$. Then, based on \Cref{lemma:rule_PITI_SN,col:piti_f_bilinear} and \eqref{eq:B_matrix_blti_nl_proof}, the resulting Koopman model $\Sigma^{\text{BLTI}}_{3,1}$ is a bilinear representation:
\begin{subequations}
\begin{align}
        \dot{z}_{3,1}&=\Ap_{3,1}z_{3,1} + \Bz_{3,1}z_{3,1}\bar{u}_{3,1}\\
        \bar{y}_{3,1}&=\Cp_{3,1}z_{3,1}
\end{align}
    \end{subequations}
where $\bar{u}_{3,1}=\bar{u}_{2,1}=u$ and, as $\Bz_{2,1}=0_{n_{\mr z , 2 ,1} \times n_{\mr z , 2 ,1}}$,
\begin{equation*}
    \Ap_{3,1} = \begin{bmatrix}
        0 & & \\
          & \Ap_{2,1} & \\
          & & {}^2\Ap_{2,1}
    \end{bmatrix} \qquad \text{and} \qquad \Bz_{3,1}=\begin{bmatrix}
        0 & & \\
        \Bzct_{2,1} & 0 & \\
        & {}^2 \Bzct_{2,1} & 0
    \end{bmatrix}.
\end{equation*}
where ${}^2 \Ap_{2,1} = \Ap_{2,1} \otimes I_{n_{\mr z , 2 , 1}} + I_{n_{\mr z , 2 , 1}} \otimes \Ap_{2,1}$ and ${}^2 \Bzct_{2,1}$ is similarly calculated. 
The output matrix and the state vector are given by
\begin{equation*}
    \Cp_{3,1} = \begin{bmatrix}
        \gamma_{3,1,0} & \gamma_{3,1,1}\Cp_{2,1} & \gamma_{3,1,2} \Cp_{2,1}^{(2)}
    \end{bmatrix} \qquad \text{and} \qquad z_{3,1} = \begin{bmatrix}
        1 & z^\top_{2,1} & \left(z^{(2)}_{2,1}\right)^\top
    \end{bmatrix}^\top.
\end{equation*}

\subsubsection{Branch 2: absorbing $\Sigma_{3,2}^\text{LTI}$}
As the next step, we embed the interconnection between $\Sigma^{\text{BLTI}}_{2,2}$ and $\Sigma^{\text{LTI}}_{3,2}$. Following  \Cref{lemma:rule_PITI_LD,corr:rule_PITI_LD}, this leads to the bilinear Koopman model $\Sigma^{\text{BLTI}}_{3,2}$:
\begin{subequations}
\begin{align}
        \dot{z}_{3,2}&=\Ap_{3,2}z_{3,2} + \Bzct_{3,2}\bar{u}_{3,2}\\
        \bar{y}_{3,2} &= \Cp_{3,2}z_{3,2}
\end{align}
    \end{subequations}
where $z_{3,2} = \begin{bmatrix}
    z^\top_{2,2} & x^\top_{3,2}
\end{bmatrix}^\top$, $\bar{u}_{3,2}=\bar{u}_{2,2}=u$, and $\bar{y}_{3,2}=y_{3,2}$. The state, input and output matrices are given by:
\begin{equation*}
    \begin{split}
        \Ap_{3,2} &= \begin{bmatrix}
            \Ap_{2,2} & 0 \\ \Bl_{3,2}\Cp_{2,2} & \Al_{3,2}
        \end{bmatrix}, \qquad \Bzct_{3,2} = \begin{bmatrix}
            \Bzct_{2,2} \\ 0
        \end{bmatrix}, \\
        \Cp_{3,2} &= \begin{bmatrix}
            0 & \Cl_{3,2}
        \end{bmatrix}.
    \end{split}
\end{equation*}
The resulting block-chain structure is given in \Cref{fig:SISO_example_steps2}.

\begin{figure}[htbp]
    \centering
    
    \begin{subfigure}{\textwidth}
        \centering
        \includegraphics[width=0.8\textwidth]{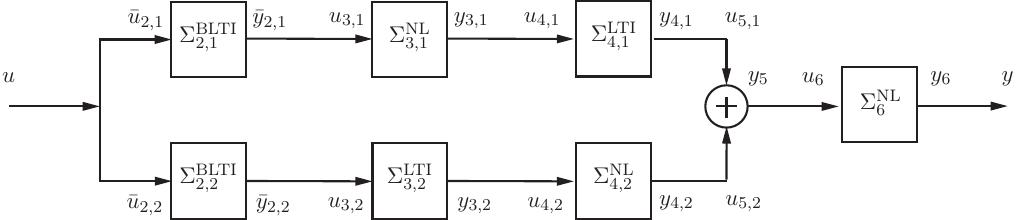}
        \vspace{0.5em}
        \caption{Embedding $\Sigma^{\text{LTI}}_1$ and the IJ.}
        \label{fig:SISO_example_steps1}
    \end{subfigure}
    
    \vspace{0.5em} 
    
    \begin{subfigure}[t]{0.54\textwidth}
        \centering
        \includegraphics[width=0.95\textwidth]{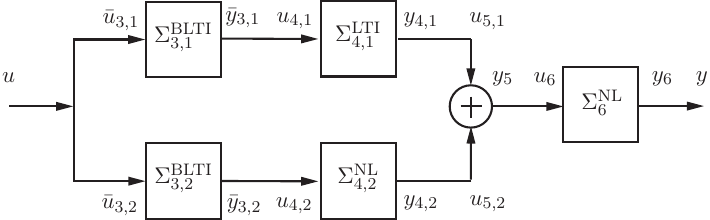}
        \vspace{0.5em}
        \caption{Embedding $\Sigma^{\text{NL}}_{3,1}$ and $\Sigma^{\text{LTI}}_{3,2}$.}
        \label{fig:SISO_example_steps2}
    \end{subfigure}
    \hspace{0.02\textwidth}
    \begin{subfigure}[t]{0.42\textwidth}
        \centering
        \includegraphics[width=0.9\textwidth]{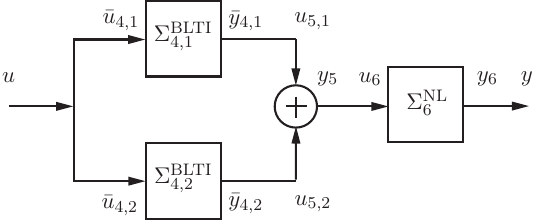}
        \vspace{0.5em}
        \caption{Embedding $\Sigma^{\text{LTI}}_{4,1}$ and $\Sigma^{\text{NL}}_{4,2}$.}
        \label{fig:SISO_example_steps3}
    \end{subfigure}
    
    \vspace{0.5em} 
    
    \begin{subfigure}[t]{0.58\textwidth}
        \centering
        \includegraphics[width=0.5\textwidth]{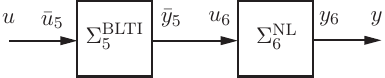}
        \vspace{0.5em}
        \caption{Embedding the OJ.}
        \label{fig:SISO_example_steps4}
    \end{subfigure}
    \hfill
    \begin{subfigure}[t]{0.38\textwidth}
        \centering
        \includegraphics[width=0.45\textwidth]{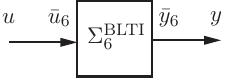}
        \vspace{0.5em}
        \caption{Embedding $\Sigma^{\text{NL}}_{6}$.}
        \label{fig:SISO_example_steps5}
    \end{subfigure}
    
    \caption{Visual representation of the steps taken to embed the SISO nonlinear blockchain system into a BLTI Koopman embedding.}
    \label{fig:SISO_example_steps}
\end{figure}
\subsubsection{Branch 1: absorbing $\Sigma_{4,1}^\text{LTI}$}
Similar to the previous step, to embed the interconnection between $\Sigma^\mathrm{BLTI}_{3,1}$ and $\Sigma^\mathrm{LTI}_{4,1}$, we use \Cref{lemma:rule_PITI_LD,corr:rule_PITI_LD} to obtain the bilinear Koopman model $\Sigma^{\text{BLTI}}_{4,1}$:
\begin{subequations}
\begin{align}
        \dot{z}_{4,1}&=\Ap_{4,1}z_{4,1} + \Bz_{4,1}z_{4,1}\bar{u}_{4,1}\\
        \bar{y}_{4,1} &= \Cp_{4,1}z_{4,1}
\end{align}
    \end{subequations}
where the state is $z_{4,1} = \begin{bmatrix}
    z^\top_{3,1} & x^\top_{4,1}
\end{bmatrix}^\top$, the input is $\bar{u}_{4,1}=\bar{u}_{3,1}=u$, and the output is $\bar{y}_{4,1}=y_{4,1}$. The matrices are given by:
\begin{equation*}
    \begin{split}
        \Ap_{4,1} &= \begin{bmatrix}
            \Ap_{3,1} & 0 \\ \Bl_{4,1}\Cp_{3,1} & \Al_{4,1}
        \end{bmatrix}, \qquad \Bz_{4,1} = \begin{bmatrix}
            \Bz_{3,1} & 0 \\ 0 & 0
        \end{bmatrix}, \\
        \Cp_{4,1} &= \begin{bmatrix}
            0 & \Cl_{4,1}
        \end{bmatrix}.
    \end{split}
\end{equation*}
\subsubsection{Branch 2: absorbing $f_{4,2}$}\label{sec:siso_example_LTI1_f42}
Next, the embedding of $\Sigma^{\text{BLTI}}_{3,2}$ followed by the nonlinear block $\Sigma^{\text{NL}}_{4,2}$, described by:
\begin{equation*}
    y_{4,2}=f_{4,2}(\bar{y}_{3,2})=\gamma_{4,2,0}+\gamma_{4,2,1}\Cp_{3,2}z_{3,2} + \gamma_{4,2,2}\Cp^{(2)}_{3,2}z_{3,2}^{(2)},
\end{equation*}
is processed.
As this interconnection is of the same type as the one discussed in \Cref{sec:siso_example_LTI1_f1}, we simply give the bilinear Koopman model $\Sigma^{\text{BLTI}}_{4,2}$:
\begin{subequations}
\begin{align}
        \dot{z}_{4,2}&=\Ap_{4,2}z_{4,2} + \Bz_{4,2}z_{4,2}\bar{u}_{4,2}\\
        \bar{y}_{4,2}&=\Cp_{4,2}z_{4,2}
 \end{align}
    \end{subequations}
where $\bar{u}_{4,2}=\bar{u}_{3,2}=u$ and, as $\Bz_{3,2}=0_{n_{\mr z , 3, 2}}$,
\begin{equation*}
    \Ap_{4,2} = \begin{bmatrix}
        0 & & \\
          & \Ap_{3,2} & \\
          & & {}^2\Ap_{3,2}
    \end{bmatrix} \qquad \text{and} \qquad \Bz_{4,2}=\begin{bmatrix}
        0 & & \\
        \Bzct_{3,2} & 0 & \\
        & {}^2 \Bzct_{3,2} & 0
    \end{bmatrix},
\end{equation*}
where ${}^2 \Ap_{3,2} = \Ap_{3,2} \otimes I_{n_{\mr z , 3 , 2}} + I_{n_{\mr z , 3 , 2}} \otimes \Ap_{3,2}$ and ${}^2 \Bzct_{3,2}$ is similarly defined. Finally, the output matrix and state vector are given by:
\begin{equation*}
    \Cp_{4,2} = \begin{bmatrix}
        \gamma_{4,2,0} & \gamma_{4,2,1}\Cp_{3,2} & \gamma_{4,2,2} \Cp_{3,2}^{(2)}
    \end{bmatrix} \qquad \text{and} \qquad z_{4,2} = \begin{bmatrix}
        1 & z^\top_{3,2} & \left(z^{(2)}_{3,2}\right)^\top
    \end{bmatrix}^\top.
\end{equation*}
The resulting block-chain structure is given in \Cref{fig:SISO_example_steps3}.

\subsubsection{Absorbing the output junction}
Next, we embed the two bilinear blocks $\Sigma^{\text{BLTI}}_{4,1}$ and $\Sigma^{\text{BLTI}}_{4,2}$ in a single BLTI block in term of the output junction. Noticing that $\bar{u}_{4,1}=\bar{u}_{4,2}=u$ and using the results in \Cref{lemma:rule_PITI_OJ,corr:OJ_OP7}, the resulting dynamics of the bilinear Koopman model $\Sigma^{\text{BLTI}}_5$ are given by:
\begin{subequations}
\begin{align}
        \dot{z}_5 &= \Ap_5 z_5 + \Bz_5 z_5 \bar{u}_5\\
        \bar{y}_5 &=\Cp_5 z_5 
\end{align}
    \end{subequations}
where $z_5 = \begin{bmatrix}
    z^\top_{4,1} & z^\top_{4,2}
\end{bmatrix}^\top$, $\bar{u}_5=u$, and $\bar{y}_5=y_5$ and the state, input, and output matrices are
\begin{equation*}
    \begin{split}
        \Ap_5 &= \begin{bmatrix}
            \Ap_{4,1} & \\ & \Ap_{4,2}
        \end{bmatrix}, \qquad \Bz_5=\begin{bmatrix}
            \Bz_{4,1} & \\ & \Bz_{4,2}
        \end{bmatrix},\\
        \Cp_5 &= \begin{bmatrix}
            \Cp_{4,1} & \Cp_{4,2}
        \end{bmatrix}.
    \end{split}
\end{equation*}
Also, note that $\Bzct_5=0_{n_{\mr z, 5}\times 1}$. The resulting block-chain structure is given in \Cref{fig:SISO_example_steps4}.

\subsubsection{Absorbing $f_{6}$}
The final step is to embed the series interconnection of the bilinear system $\Sigma^{\text{BLTI}}_{5}$ and the nonlinear block $\Sigma^{\text{NL}}_6$, described by:
\begin{equation*}
    y_6=f_6(\bar{y}_5)=\gamma_6 + \gamma_{6,1}\Cp_5 z_5 + \gamma_{6,2} \Cp_5^{(2)} z^{(2)}_5.
\end{equation*}
The derivation follows the same reasoning as detailed in \Cref{sec:siso_example_LTI1_f1,sec:siso_example_LTI1_f42}. Based on \Cref{lemma:rule_PITI_SN,col:piti_f_bilinear}, we obtain
\begin{subequations}
\begin{align}
        \dot{z}_6&= \Ap_6 z_6 + \Bz_6 z_6 \bar{u}_6\\
        \bar{y}_6&= \Cp_6 z_6
\end{align}
    \end{subequations}
with $z_6=[\
    1 \ z_5^\top \ (z^{(2)}_5)^\top\
]^\top$, $\bar{u}_6=u$, and $\bar{y}=y$. The state, input, and output matrices are given by:
\begin{equation*}
    \begin{split}
        \Ap_6 &= \begin{bmatrix}
            0 & & \\ 
            & \Ap_5 & \\
            & & {}^2 \Ap_5
        \end{bmatrix}, \qquad  \Bz_6 = \begin{bmatrix}
            0 & & \\
             & \Bz_5 & \\
            &  & {}^2 \Bz_5
        \end{bmatrix}, \\
        \Cp_6 &= \begin{bmatrix}
            \gamma_{6,0} & \gamma_{6,1}\Cp_5 & \gamma_{6,2}\Cp^{(2)}_5
        \end{bmatrix}.
    \end{split}
\end{equation*}
as, using \eqref{eq:B_matrix_blti_nl_proof}, we have that $\Bzct_5 = 0_{n_{\mr z, 5}\times 1}$ and we can also compute ${}^2\Ap_5=\Ap_{5} \otimes I_{n_{\mr z , 5}} + I_{n_{\mr z ,5}} \otimes \Ap_{5}$ and ${}^2\Bz_5=\Bz_5 \otimes I_{n_{\mr z , 5}} + I_{n_{\mr z ,5}} \otimes \Bz_5$. This results in \Cref{fig:SISO_example_steps2}, completing the process.

\subsubsection{Final model}

In the previous subsections, we have shown that the considered SISO nonlinear system, described by \Cref{fig:SISO_example_1}, can be exactly embedded into a BLTI model $\Sigma^{\text{BLTI}}_6$ with the lifted dynamics described by:
\begin{equation}\label{eq:SISO_example_bilinear_full}
\begin{split}
\dot{z}_6 &=\Ap_6 z_6 + \Bz_6 z_6 u\\
y&= \Cp_6 z_6
\end{split}
\end{equation}
with $n_\mathrm{u}=1$ and $n_\mathrm{y}=1$. The resulting lifted state $z_6$ is of dimension $n_{\mathrm{z},6}=931$, even though the state dimension of each LTI block is $n_\mathrm{x}=2$ and the maximum polynomial power in the nonlinear blocks in this example is $p=2$. One of the reasons for this, as discussed in \Cref{sec:final_models_MIMO_example}, is the high number of duplicate states resulting from the Kronecker products. Furthermore, at multiple embedding steps, constants are introduced in the state vector, e.g., $z_6=[\ 
    1 \ z_5^\top \ (z^{(2)}_5)^\top \ ]^\top$.
To remove the duplicate states, we apply the linear transformation $z=Tz_6$ with $z_6=T^\dagger z$, where $T\in\mathbb{R}^{n_\mr z \times n_{\mr z ,6}}$ is the transformation matrix that selects the unique elements and $T^\dagger$ is its inverse. This gives the final model $\Sigma^{\text{BLTI}}$ with:
\begin{equation}\label{eq:example_siso_reduced_blti}
\begin{split}
\dot{z}&= \Ap z +  \Bz z u\\
y&= \Cp z
\end{split}
\end{equation}
where $ A= T\Ap_6T^\dagger$, $ \Bz=T\Bz_6T^\dagger$, $ C= \Cp_6 T^\dagger$, and reduced state dimension of $n_{\mr z}=103$.
\par 
\begin{figure}[ht!]
\begin{center}   
\includegraphics[width=0.7\textwidth]{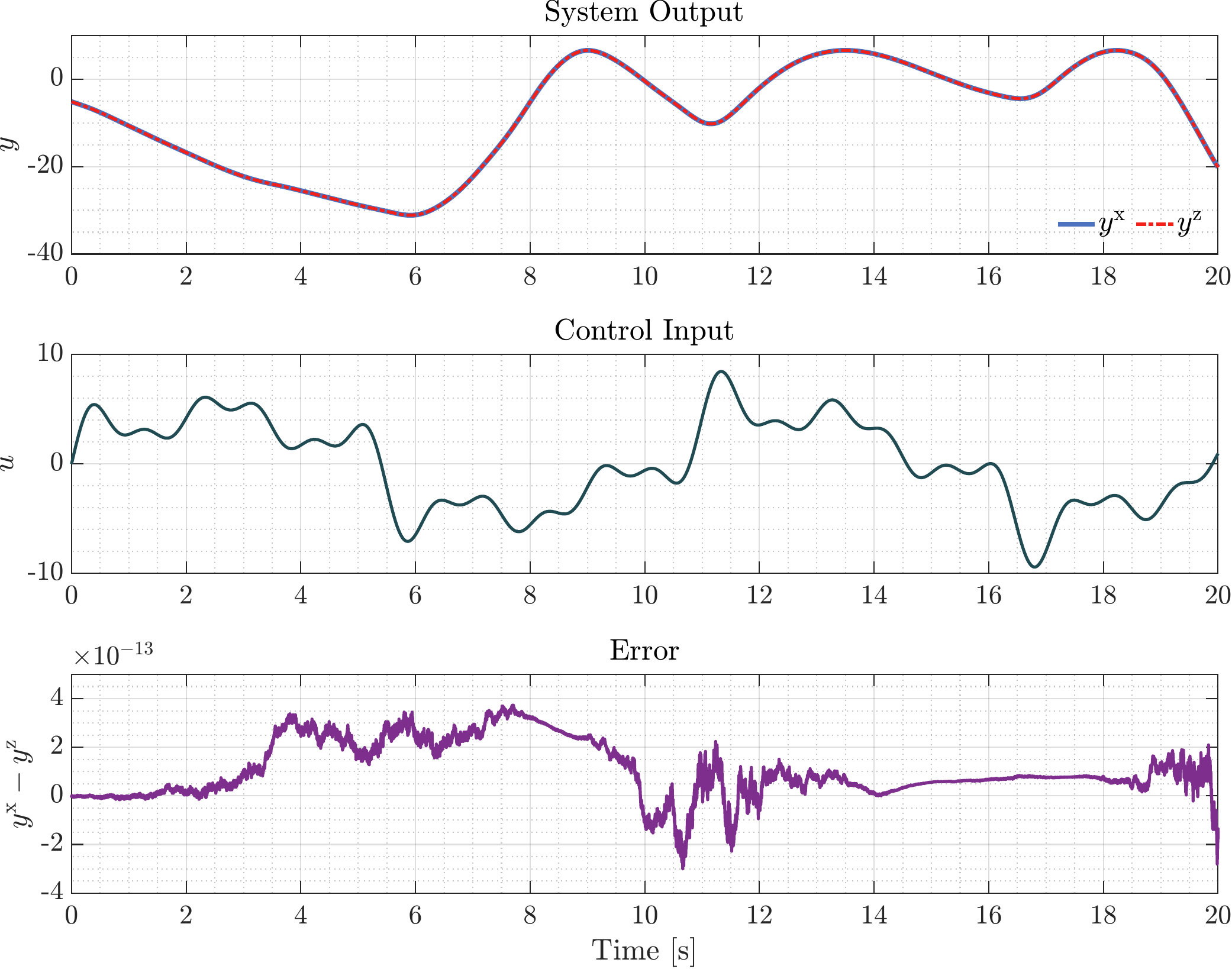}
\caption{Simulated output response ($y^\mathrm{x}$) of the SISO NL block chain system depicted in \Cref{fig:SISO_example_1} and the response ($y^\mathrm{z}$) of the embedding process provided BLTI Koopman model $\Sigma^{\text{BLTI}}$  (top plot) for the multisine input signal $u$ (middle plot) and the difference of the obtained responses (bottom plot).}
\label{fig:SISO_results}
\end{center}
\end{figure}
To show that the obtained model is an exact representation of the original NL block chain model, simulation responses of the original nonlinear system and of the reduced Koopman BLTI model $\Sigma^{\text{BLTI}}$ under a multisine input are given in \Cref{fig:SISO_results}. The input $u=\sum_iA_i\sin(2\pi f^{\text{Hz}}_i t)$ is a sum of 6 sinusoids, with frequencies from an equidistant grid between $0.1$ and $1\text{ Hz}$, and various amplitudes. The initial condition of the states of the LTI blocks is chosen as $x_{1,0}=x_{3,2,0}=x_{4,1,0}=[1 \;\; 1]^\top$, while  $z_0 = Tz_{6,0}$ is based on the construction of $z_6$. The numerical integration method used to obtain the responses is Runge Kutta 4 with a step size of $\delta t=10^{-4}$s. \Cref{fig:SISO_results} shows that the error between the simulated output $y^\mathrm{x}$ of the nonlinear system and the output $y^\mathrm{z}$ of the obtained BILTI Koopman model  $\Sigma^{\text{BLTI}}$ is in the order of $10^{-13}$, which is close to numerical precision. This shows that \eqref{eq:example_siso_reduced_blti} is an exact embedding.
\vspace{-.1cm}
\section{Conclusions}
\label{sec:conclusions}

The present paper treats the problem of deriving exact and ~finite-dimensional Koopman models. Starting from a nonlinear system that is represented as a network of linear and nonlinear blocks (the Wiener-Hammerstein system and its different configurations are well-known examples), a Koopman model with constant state and output matrices and polynomial input structure is obtained by exploiting the properties of the Kronecker product of the states. Moreover, if the linear blocks do not have feedthrough terms, an exact bilinear model can be derived. This is a particularly exciting result, as such a structure has been heavily applied in the Koopman-form-based control of nonlinear systems with great results in practice, and prior exact derivations of bilinear Koopman models were based on conditions much more difficult to satisfy (see \eqref{eq:bilinear_state_condition}). Furthermore, we provide an algorithm to directly compute the analytic form of these finite Koopman models requiring no data or approximations compared to other methods in the literature.
Examples both for the PITI and BLTI forms have been discussed to validate the technique and the resulting models and to demonstrate that the nonlinear behavior is exactly captured. 

\appendix
\section{Lemmas and proofs}
\subsection{Kronecker product properties}\label{app:A1}
We list here several useful properties of the Kronecker product for completeness as they are used later in the proofs. \par 
We start with the \emph{Mixed Product Property} (MPP) of the Kronecker product. That is, given matrices $A,B,C,D$ with appropirate dimensions such that $AC$, $BD$ can be computed, then, as detailed in Proposition 7.1.6 in \cite{MatrixMath}, it holds that:
\begin{equation}\label{eq:mixed_prod_property}
    (A\otimes B)(C\otimes D) = (AC) \otimes (BD).
\end{equation}
 \par
A second property detailed in Fact 7.4.1 in \cite{MatrixMath} is that a Kronecker product of two vectors $x,y\in\mathbb{R}^{n}$ can be alternatively described as:
\begin{equation}\label{eq:expanded_vector_prod}
    x \otimes y = \left(x \otimes I_n \right)y = \left(I_n \otimes y \right)x.
\end{equation}
 \par
Finally, the Kronecker power of the product of two matrices $A\in\mathbb{R}^{n\times m}$ and $B\in\mathbb{R}^{m \times l}$ can be expanded as:
\begin{equation}\label{eq:decoupled_kron_prod}
    \left(AB\right)^{(k)}=A^{(k)}B^{(k)}
\end{equation}
as noted in Fact 7.4.10 in \cite{MatrixMath}. Note that \eqref{eq:decoupled_kron_prod} also holds if $B$ is  a vector of dimension $\mathbb{R}^m$.
\subsection{The Kronecker gradient}\label{app:A2}
Let  $x\in\mathbb{R}^{n_\mr x}$ and $x^{(i)}$ denote the $i^\mathrm{th}$ Kronecker power, i.e., $x^{(i)} = \underbrace{x \otimes \dots \otimes x}_{i\text{ times}}$, with $x^{(1)}=x$ and $x^{(0)}=1\in\mathbb{R}$. Then, the following Lemma holds.
\begin{lemma}\label{lemma:appendix_lemma_1} The Jacobian of the $i^\mathrm{th}$ Kronecker power $x^{(i)}$ with $x\in\mathbb{R}^{n_\mathrm{x}}$ is
\begin{equation}
\frac{\partial x^{(i)}}{\partial x}=\sum^{i-1}_{k=0}x^{(k)} \otimes I_{n_\mr x} \otimes x^{(i-k-1)}.
\end{equation}
\end{lemma}
\begin{proof}
	We use the property given in \cite{magnus_differential}:
\begin{equation}
	\dif \left(x \otimes x\right) = \dif x \otimes x + x \otimes\dif x
\end{equation}
where $\dif x$~is the differential of $x$. For a function $f:\mathbb{R}^{n_\mr x}\rightarrow \mathbb{R}^{n_\mr f}$, the connection between the differential and the first derivative (Jacobian) is given by:
\begin{equation}\label{eq:connection_differential_derivative}
    \dif f = \frac{\partial f}{\partial x}\dif x
\end{equation}
see Theorem 18.1 in \cite{magnus_econometrics}. We accomplish the proof by induction.

{\bf Case $i=1$:} This case is straightforward with $\frac{\partial x^{(1)}}{\partial x}=\frac{\partial x}{\partial x}=I_{n_\mathrm{x}}$.

{\bf Case $i=2$:} Start with $\dif x^{(2)} = \dif \left(x \otimes x\right) = \dif x \otimes x + x \otimes\dif x$ and  
	 use property \eqref{eq:expanded_vector_prod} to write:
	\begin{align}
		\dif x ^{(2)}&=\left(I_{n_\mr x}\otimes x \right) \dif x + \left(x \otimes I_{n_\mr x}\right) \dif x \\
        &=\left(I_{n_\mathrm{x}}\otimes x + x \otimes I_{n_\mathrm{x}}\right)\dif x \notag 
        \end{align}
	Then, using \eqref{eq:connection_differential_derivative}, as $\dif x^{(2)}=\frac{\partial x^{(2)}}{\partial x}$, it follows that:
	\begin{equation}
		\frac{\partial x^{(2)}}{\partial x}=I_{n_\mathrm{x}}\otimes x + x \otimes I_{n_\mathrm{x}}.
	\end{equation}

{\bf Case $i=3$:} As $x^{(3)}= x^{(2)}\otimes x$, based on the previous derivations, it holds that:
	\begin{align}
		\dif x^{(3)} &= \dif x^{(2)} \otimes x + x^{(2)}\otimes \dif x,\\
		&=\left(\dif x \otimes x + x \otimes \dif x\right)\otimes x + x^{(2)} \otimes \dif x, \notag\\
        &=\left(I_{n_\mathrm{x}}\otimes x + x \otimes I_{n_\mathrm{x}} \right)\dif x \otimes x + (x^{(2)}\otimes I_{n_{\mr x}})\dif x,\notag\\
        &= \left(I_{n_\mathrm{x}}\otimes x\otimes x + x \otimes I_{n_\mathrm{x}} \otimes x + x \otimes x \otimes I_{n_\mathrm{x}}\right) \dif x. \notag
		\end{align}
   As $\dif x^{(3)} = \frac{\partial x^{(3)}}{\partial x}\dif x$, it holds that:
	\begin{equation}
			\frac{\partial x^{(3)}}{\partial x}=I_{n_\mathrm{x}}\otimes x^{(2)} + x \otimes I_{n_\mathrm{x}} \otimes x + x^{(2)} \otimes I_{n_\mathrm{x}}.
	\end{equation}

{\bf Case $i+1$: }
Let \begin{equation}
	\frac{\partial x ^{(i)}}{\partial x}=\sum^{i-1}_{k=0} x^{(k)} \otimes I_{n_\mathrm{x}} \otimes x^{(i-k-1)}.
\end{equation}
We need to prove that:
\begin{equation}\label{eq:appendix_dx_induxtion_last_step}
	\frac{\partial x ^{(i+1)}}{\partial x}=\sum^{i}_{k=0} x^{(k)} \otimes I_{n_\mathrm{x}} \otimes x^{(i-k)}.
\end{equation}
We start with $x^{(i+1)}=x^{(i)}\otimes x$ and use $\dif x^{(i)}=\frac{\partial x^{(i)}}{\partial x}$. Then:
\begin{align}
		\dif x^{(i+1)}&= \dif\ (x^{(i)} \otimes x) 
		=\dif x^{(i)} \otimes x + x^{(i)} \otimes \dif x \\
		&= \sum^{i-1}_{k=0} x^{(k)} \otimes \dif x \otimes x^{(i-k-1)} \otimes x  + x^{(i)}\otimes I_{n_\mr x}  \notag \\
        &= \sum^{i}_{k=0} x^{(k)} \otimes \dif x \otimes x^{(i-k)}  \notag
	\end{align}
such that we obtain \eqref{eq:appendix_dx_induxtion_last_step}. This, by induction, proves \Cref{lemma:appendix_lemma_1}.
\end{proof}
\subsection{The Kronecker gradient product rules}
\begin{lemma}\label{lemma:appendix_lemma_2} Let $x\in\mathbb{R}^{n_\mr x}$, $I_{n_\mr x}\in\mathbb{R}^{n_\mr x \times n_\mr x}$ and $A\in\mathbb{R}^{n_\mr x \times n_\mr x}$. It holds that
	
    \begin{equation}
    \sum^{i-1}_{k=0}\left(x^{(k)} \otimes I_{n_{\mr x}}\otimes x^{(i-k-1)}\right)Ax=\sum^{i-1}_{k=0} x^{(k)} \otimes Ax \otimes x^{(i-k-1)}.
    \end{equation} 
\end{lemma}
\begin{proof}
	Each element of the sum on the left is $\left(x^{(a)}\otimes I_{n_\mr x}\otimes x^{(b)}\right)Ax$ with $a=k,\;b=i-k-1$ where  $Ax\in\mathbb{R}^{n_\mr x}$ is a vector.  
Then using the identities in \Cref{app:A1}
and $Ax\otimes 1 = Ax$, we have:
		\begin{align}
		    \left(x^{(a)}\otimes I_{n_x}\otimes x^{(b)}\right)Ax&= \;\left(x^{(a)}\otimes I_{n_x}\otimes x^{(b)}\right)\left(Ax\otimes 1\right)\\
            &\stackrel{\mathmakebox[\widthof{=}]{\eqref{eq:mixed_prod_property}}}{=}\;\left(\left(x^{(a)} \otimes I_{n_\mr x}\right) Ax\right)\otimes x^{(b)}\notag\\
            &\stackrel{\mathmakebox[\widthof{=}]{\eqref{eq:expanded_vector_prod}}}{=}\; x^{(a)}\otimes Ax \otimes x^{(b)} \notag
		\end{align}
This holds for all elements in the sum, thus \Cref{lemma:appendix_lemma_2} holds. 
\end{proof}
\begin{lemma}\label{lemma:appendix_lemma_3}
	Let $x\in\mathbb{R}^{n_\mr x}$ and $A\in\mathbb{R}^{n_\mr x \times n_\mr x}$. It holds that:
   \begin{equation}
        \sum^{i-1}_{k=0} x^{(k)} \otimes Ax \otimes x^{(i-k-1)} = \sum^{i-1}_{k=0}\left(I_{n_\mr x}^{(k)} \otimes A \otimes I_{n_\mr x}^{(i-k-1)}\right) x^{(i)} .
    \end{equation}
\end{lemma}

\begin{proof}
	Each element of the sum on the left is $x^{(a)}\otimes Ax \otimes x^{(b)}$ with $a=k,\; b=i-k-1$. Using \eqref{eq:decoupled_kron_prod}, we have that $x^{(a)}=\left(I_{n_\mr x} x\right)^{(a)}=I_{n_\mr x}^{(a)}x^{(a)}$. Then:
	\begin{align}
		x^{(a)}\otimes Ax \otimes x^{(b)} &=\; \left(I_{n_x}^{(a)}x^{(a)}\otimes Ax\right)\otimes x^{(b)}\\
		&\stackrel{\mathmakebox[\widthof{=}]{\eqref{eq:mixed_prod_property}}}{=}\;\left(\left(I_{n_x}^{(a)}\otimes A\right)x^{(a+1)}\right)\otimes x^{(b)}\notag\\
		&\stackrel{\mathmakebox[\widthof{=}]{\eqref{eq:decoupled_kron_prod}}}{=}\;\left(\left(I_{n_x}^{(a)}\otimes A\right)x^{(a+1)}\right)\otimes\left(I^{(b)}_{n_x}x^{(b)}\right)\notag\\
		&\stackrel{\mathmakebox[\widthof{=}]{\eqref{eq:mixed_prod_property}}}{=}\;\;\left(I^{(a)}_{n_x}\otimes A\otimes I^{(b)}_{n_x}\right)x^{(a+b+1)}\notag
	\end{align}
    and $a+b+1 = i$. As this holds for all elements of the sum,  \Cref{lemma:appendix_lemma_3} holds.
\end{proof}

Overall, \Cref{lemma:appendix_lemma_1,lemma:appendix_lemma_2,lemma:appendix_lemma_3} show that, for $x\in\mathbb{R}^{n_\mr x}$ and $A\in\mathbb{R}^{n_\mr x \times n_\mr x}$:
\begin{equation}\label{eq:kron_gradient_Ax}
    \frac{\partial x^{(i)}}{\partial x}Ax = \underbrace{\left(\sum^{i-1}_{k=0}I_{n_\mr x}^{(k)}\otimes A \otimes I^{(i-k-1)}_{n_\mr x}\right)}_{{}^iA}x^{(i)}.
\end{equation}

\par 
Moreover, the multiplication of the gradient of $x^{(i)}$ with $B\in\mathbb{R}^{n_\mr x}$, which can be seen as a column of $\bar{B}\in\mathbb{R}^{n_\mathrm{x}\times n _\mr u}$, is similar to \eqref{eq:kron_gradient_Ax}: 
\begin{lemma}\label{lemma:appendix_lemma_partial_B}
For $x\in\mathbb{R}^{n_\mr x}$ and $B\in\mathbb{R}^{n_\mr x}$, the product between the gradient of $x^{(i)}$ and $B$ can be expressed as:
\begin{equation}\label{eq:kron_gradient_partial_B}
    \frac{\partial x^{(i)}}{\partial x}B =\underbrace{\left(\sum^{i-1}_{k=0}I_{n_{\mr x}}^{(k)}\otimes  B \otimes I_{n_{\mr x}}^{(i-k-1)}\right)}_{{}^i B}x^{(i-1)}. 
\end{equation}
\end{lemma}

\begin{proof}
First, based on \Cref{lemma:appendix_lemma_1}:
\begin{equation}
		\frac{\partial x^{(i)}}{\partial x}B  =\left(\sum^{i-1}_{k=0}x^{(k)}\otimes I_{n_{\mr x}} \otimes  x^{(i-k-1)}\right)B.
\end{equation}
Next, based on \Cref{lemma:appendix_lemma_2}, for each term of the sum, it holds that:

\begin{equation}
		 \left((x^{(a)}\otimes I_{n_{\mr x}} \otimes  x^{(b)}\right)B =\;x^{(a)} \otimes B \otimes x^{(b)}.
\end{equation} 

Next, we factorize the Kronecker powers of $x$ in the term $x^{(a)}\otimes B \otimes x^{(b)}$. Note that we denote here the multiplication with a  scalar by $\cdot$ and we obtain:
	\begin{align}
			x^{(a)}\otimes B \otimes x^{(b)}\; &\stackrel{\mathmakebox[\widthof{=}]{\eqref{eq:decoupled_kron_prod}}}{=}\;\left(I^{(a)}_{n_x}x^{(a)}\otimes B \cdot 1\right)\otimes {x^{(b)}}\\
			&\stackrel{\mathmakebox[\widthof{=}]{\eqref{eq:mixed_prod_property}}}{=}\;\left(I^{(a)}_{n_x}\otimes B\right)\left(x^{(a)}\otimes 1\right)\otimes x^{(b)}\notag\\
			&\stackrel{\mathmakebox[\widthof{=}]{\eqref{eq:decoupled_kron_prod}}}{=}\;\left(I^{(a)}_{n_x} \otimes B\right)x^{(a)} \otimes \left(I^{(b)}_{n_x}x^{(b)}\right)\notag\\
			&\stackrel{\mathmakebox[\widthof{=}]{\eqref{eq:mixed_prod_property}}}{=}\;\left(I^{(a)}_{n_x}\otimes B \otimes I^{(b)}_{n_x}\right)x^{(a+b)}\notag
		\end{align}
	with $a+b=i-1$. This again holds for every term of the sum, concluding the proof of \Cref{lemma:appendix_lemma_partial_B}.
\end{proof}
\subsection{Proof of \Cref{lemma:rule_PITI_SN}}\label{apx:PITI_f_to_PITI}

We start with the decomposition of the function $f_i$ in terms of $f_i(\bar{y}_{i-1}) = W_ig_i(V_i^\top \bar{y}_{i-1})$. We describe $g_{i,e}$ as follows, with $e\in\{1,\dots,r_i\}$, where $r_i$ is the decoupling order of $f_i$ and $p_i$ is the maximum of the monomial orders in the decomposition 
of $f_i$: 
\begin{align}
		g_{i,e}(\sigma_{i,e})&= \gamma_{i,e,0} + \gamma_{i,e,1}\sigma_{i,e} + \dots + \gamma_{i,e,p_i}\sigma_{i,e}^{p_i}\\
        &=\gamma_{i,e,0} + \gamma_{i,e,1}\left(v_{i,e}\bar{y}_{i-1}\right) + \dots+\gamma_{i,e,p_i}\left(v_{i,e}\bar{y}_{i-1}\right)^{p_i}\notag\\
		&= \gamma_{i,e,0} + \gamma_{i,e,1} \left(\tilde{v}_{i,e} z_{i-1} + \hat{v}_{i,e}\bar{u}_{i-1}\right)+ \dots + \gamma_{i,e,p_i} \left(\tilde{v}_{i,e} z_{i-1} + \hat{v}_{i,e}\bar{u}_{i-1}\right)^{p_i}\notag
	\end{align}
with $\sigma_{i,e}=v_{i,e}\bar{y}_{i-1}$, where $v_{i,e}$ is the $e\textsuperscript{th}$ column of $V^\top_i$, such that:
\begin{equation*}
    \tilde{v}_{i,e}=v_{i,e}\Cp _{i-1} \qquad \text{and} \qquad \hat{v}_{i,e}=v_{i,e}\Py _{i-1}(z_{i-1})\Ry _{i-1}(\bar{u}_{i-1}).
\end{equation*}
We continue expanding $g_{i,e}(\sigma_{i,e})$ as follows:
\begin{multline}\label{eq:initial_gie_decomp}
        g_{i,e}(\sigma_{i,e})=\underbrace{\begin{bmatrix}
			\gamma_{i,e,0} \; \gamma_{i,e,1} \tilde{v}_{i,e} \;  \cdots \; \gamma_{i,e,p_i}\tilde{v}^{(p_i)}_{i,e}
		\end{bmatrix}}_{\Gamma_{i,e}} \underbrace{\begin{bmatrix}
1 \\z_{i-1}  \\ \vdots \\ z^{(p_i)}_{i-1} \end{bmatrix}}_{z_i} + \\ 
\left(\sum^{p_i}_{l=1}\gamma_{i,e,l}\sum^l_{k=1}\begin{pmatrix}
	l \\ k\end{pmatrix}\tilde{v}^{(l-k)}_{i,e} z^{(l-k)}_{i-1}\hat{v}^{(k)}_{i,e} \left(I_{n_{\bar{\mathrm{u}},i-1}}\otimes \bar{u}^{(k-1)}_{i-1}\right)
\right)\bar{u}_{i-1}
    \end{multline}
where we used $\bar{u}^{(k)}_{i-1}=\left(I_{n_{\bar{\mr u},i-1}}\otimes \bar{u}_{i-1}^{(k-1)}\right)\bar{u}_{i-1}$ based on \eqref{eq:expanded_vector_prod}.
We can further expand $g_{i,e}(\sigma_{i,e})$ as:
\begin{equation}\label{eq:expanded_g_ie_decomp}
        \Gamma_{i,e} z_i + 
\left(\sum^{p_i}_{l=1}\gamma_{i,e,l}\sum^l_{k=1}\begin{pmatrix}
	l \\ k\end{pmatrix}\tilde{v}^{(l-k)}_{i,e} z^{(l-k)}_{i-1}v^{(k)}_{i,e} \Py _{i-1}^{(k)}(z_{i-1})\Ry _{i-1}^{(k)}(\bar{u}_{i-1}) \left(I_{n_{\bar{\mathrm{u}},i-1}}\otimes \bar{u}^{(k-1)}_{i-1}\right)
\right)\bar{u}_{i-1},
\end{equation}
where, based on \eqref{eq:decoupled_kron_prod}, we used:
\begin{equation*}
    \begin{split}
        \hat{v}_{i,e}^{(k)}&=\left(v_{i,e}\Py _{i-1}(z_{i-1})\Ry _{i-1}(\bar{u}_{i-1})\right)^{(k)}\\
        &=\left(v_{i,e}\Py _{i-1}(z_{i-1})\right)^{(k)}\Ry _{i-1}^{(k)}(\bar{u}_{i-1}) = v_{i,e}^{(k)}\Py _{i-1}^{(k)}(z_{i-1})\Ry _{i-1}^{(k)}(\bar{u}_{i-1}).
    \end{split}
\end{equation*}
\vspace{-.2cm}
Let 
\begin{multline}\label{eq:eq_linear_in_z_output_ie}
        \Py _{i,e}(z_{i-1})=\left[\begin{matrix}
		\gamma_{i,e,1}\tilde{v}^{(0)}_{i,e}z^{(0)}_{i-1}v^{(1)}_{i,e} \Py _{i-1}^{(1)}(z_{i-1}) & 2\gamma_{i,e,2}\tilde{v}^{(1)}_{i,e}z^{(1)}_{i-1}v^{(1)}_{i,e} \Py _{i-1}^{(1)}(z_{i-1})\end{matrix} \right. \\
		  \qquad \qquad \qquad \left. \begin{matrix} \gamma_{i,e,2}\tilde{v}^{(0)}_{i,e}z^{(0)}_{i-1}v^{(2)}_{i,e} \Py _{i-1}^{(2)}(z_{i-1}) & \dots &  \gamma_{i,e,p_i}\tilde{v}^{(0)}_{i,e}z^{(0)}_{i-1}v^{(p_i)}_{i,e} \Py _{i-1}^{(p_i)}(z_{i-1})
	\end{matrix} \right].
    \end{multline}
As the maximal Kronecker product of $z_{i-1}$ in $\Py _{i,e}(z_{i-1})$ is $z^{(p_i-1)}_{i-1}$, we define:
\begin{equation}
    \Py_i (z_i) := W_i \begin{bmatrix}
        \Py _{i,1} (z_{i-1}) \\ \vdots \\ \Py _{i,r_i}(z_{i-1})
    \end{bmatrix} \qquad \text{with} \qquad z_i=\begin{bmatrix}
        1 \\ z_{i-1} \\ \vdots \\ z_{i-1}^{(p_i)}
    \end{bmatrix}
\end{equation}
where the premultiplication with $W_i$ comes from the decomposition of $f_i$ as given by \eqref{eq:piti_f_output_nl} and $\Py _i (z_i)$ is linear in $z_i$. The latter is due to the fact that in \eqref{eq:eq_linear_in_z_output_ie}, $\Py _{i-1}(z_{i-1})$ is linear in $z_{i-1}$ and each of its Kronecker powers $\Py _{i-1}^{(\tau)}(z_{i-1})$ will be linear in $\{1, z_{i-1}, \ldots, z_{i-1}^{(\tau)}\}$. As $p_i$ is the highest Kronecker power that can occur, the resulting \eqref{eq:eq_linear_in_z_output_ie}, for all $e\in\{1,\dots,r_i\}$, will be linear in the elements of $z_i$ composed from $\{1, z_{i-1}, \ldots, z_{i-1}^{(p_i)}\}$.
The remaining terms of \eqref{eq:expanded_g_ie_decomp} are gathered to define:
\begin{equation}
    \Ry _i (\bar{u}_i):= \begin{bmatrix}
        \Ry _{i-1}(\bar{u}_{i-1}) \\
        \Ry _{i-1}(\bar{u}_{i-1}) \\
        \Ry _{i-1}^{(2)}(\bar{u}_{i-1}) \\
        \vdots \\
        \Ry _{i-1}^{(p_i)}(\bar{u}_{i-1})\left(I_{n_{\bar{\mr u},i-1}}\otimes \bar{u}^{(p_i-1)}_{i-1}\right)        
    \end{bmatrix}
\end{equation}
with $\bar{u}_i \equiv \bar{u}_{i-1}$. Finally, let $\Gamma_i = \begin{bmatrix}\Gamma^\top_{i,1} \; \cdots\; \Gamma^\top_{i,r_i}\end{bmatrix}^\top$. Then, the output equation is:
\begin{equation}
    \bar{y}_i = \underbrace{W_i \Gamma_i}_{\Cp _i}z_i + \Py _i (z_i) \Ry _i (\bar{u}_i) \bar{u}_i.
\end{equation}
Next, to derive the sate equation, we take the time derivatives of the elements of $z_i$, i.e., the time derivative of $z^{(j)}_{i-1}$.
For $j=1$, we obtain the time derivative of $z_{i-1}$ as
\begin{equation}
    \dot{z}_{i-1} = \Ap _{i-1} z_{i-1} + \Px _{i-1}(z_{i-1}) \Rx _{i-1} (\bar{u}_{i-1})\bar{u}_{i-1} .
\end{equation}
For $j\in\{2,\dots,p_i\}$, we have:
\begin{align}
        \frac{\dif }{\dif t}z^{(j)}_{i-1} &= \frac{\partial z^{(j)}_{i-1}}{\partial z_{i-1}}\Ap _{i-1} z_{i-1} + \frac{\partial z^{(j)}_{i-1}}{\partial z_{i-1}} \Px _{i-1}(z_{i-1}) \Rx _{i-1} (\bar{u}_{i-1})\bar{u}_{i-1} \\
        &= {}^j \Ap _{i-1} z^{(j)}_{i-1} + \frac{\partial z^{(j)}_{i-1}}{\partial z_{i-1}} \Px _{i-1}(z_{i-1}) \Rx _{i-1} (\bar{u}_{i-1})\bar{u}_{i-1}, \notag
    \end{align}
where ${}^j \Ap _{i-1}$ is expressed as  given by \eqref{eq:kron_gradient_Ax}, i.e., ${}^j \Ap _{i-1}=\sum^{i-1}_{k=0}I_{n_\mr x}^{(k)}\otimes A \otimes I^{(i-k-1)}_{n_\mr x}$. Note that the first element of $z_i$ is 1, so $\dot{1}=0$. To fit this relation into the state equation, we construct the state transition $0=0\cdot 1 + 0_{1\times n_{\bar{\mr u},i-1}} \bar{u}_{i-1}$. Stacking all Kronecker products, we obtain:
\begin{equation}\label{eq:proof_lemma_piti_f_results_x}
        \Ap _i = \begin{bmatrix}
            0 & & & & \\
            & \Ap _{i-1} & & & \\
            & & {}^2 \Ap _{i-1} & & \\
            & & & \ddots & \\
            & & & & {}^{p_i} \Ap _{i-1}
        \end{bmatrix}, \quad \Px _i (z_i) := \underbrace{\begin{bmatrix}
            0_{1 \times n_{\mr z , i-1}} \\
            I_{n_{\mr z, i-1}} \\ 
            \frac{\partial z^{(2)}_{i-1}}{\partial z_{i-1}}\\
            \vdots \\
            \frac{\partial z^{(p_i)}_{i-1}}{\partial z_{i-1}}
        \end{bmatrix}}_{J_{p_i}(z_{i-1})}\Px _{i-1} (z_{i-1}) ,
\end{equation}
and $\Rx _i (\bar{u}_i) := \Rx _{i-1}(\bar{u}_{i-1})$ with $\bar{u}_i \equiv \bar{u}_{i-1}$.
It can be observed that $\Px _i (z_i)$ maintains linearity in $z_i$. Through the partial derivative, the Kronecker products drop in power by one as can be seen in \Cref{lemma:appendix_lemma_1,lemma:appendix_lemma_partial_B}. Then, a multiplication between elements of $z^{(a)}_{i-1}$ with $a \leq p_i -1$ and $z_{i-1}$ will generate elements of at maximum $z^{(p_i)}_{i-1}$ in $z_i$, ensuring linearity of $\Px _i (z_i)$ in the new state $z_i$.

\subsection{Proof of \Cref{col:piti_f_bilinear}}\label{apx:PITI_f_to_BLTI}
First, $\Sigma^{\text{PITI}}_{i-1}$ is bilinear, that means:
\begin{align}
    \Px _{i-1} (z_{i-1}) \Rx _{i-1} (\bar{u}_{i-1})\bar{u}_{i-1} &= \left(\sum^{n_{\mr z , i-1}}_{j=1}\Bz _{i-1,j}z_{i-1,j} + \Bzct _{i-1,j}\right)\bar{u}_{i-1}\\
    &= \sum^{n_{\bar{\mr u},i-1}}_{k=1} \left({}_k\Bz _{i-1}z_{i-1} + {}_k \Bzct _{i-1}\right)\bar{u}_{i-1,k} \notag
    \end{align}
    where $\Px _{i-1} (z_{i-1})$ is linear in $z_{i-1}$ and $\Rx _{i-1} (\bar{u}_{i-1})$ by construction is composed from ones and zeros.
As $\bar{u}_i\equiv \bar{u}_{i-1}$, based on \eqref{eq:proof_lemma_piti_f_results_x}, we have that:
\begin{equation}
    \Px _i (z_i) \Rx _i (\bar{u}_i)\bar{u}_i := \underbrace{J_{p_i}(z_{i-1})\Px _{i-1} (z_{i-1})}_{\Px _i (z_i) } \underbrace{\Rx _{i-1} (\bar{u}_{i-1})}_{\Rx _i (\bar{u}_i)}\bar{u}_{i-1}
\end{equation}
with $J_{p_i}(z_{i-1})$ defined in \eqref{eq:proof_lemma_piti_f_results_x}. In \Cref{apx:PITI_f_to_PITI}, we have already shown that $\Px _i (z_i)$ remains linear in the new state $z_i$, while $\Rx _i (\bar{u}_i)$ will inherit that it is composed from ones and zeros, due to $\Rx _i (\bar{u}_i) \equiv \Rx _{i-1} (\bar{u}_{i-1}) $. This concludes the proof for the state equation, but it is interesting to show that, for $j\in\{2,\dots,p_i\}$, based on \Cref{lemma:appendix_lemma_1,lemma:appendix_lemma_2,lemma:appendix_lemma_3,lemma:appendix_lemma_partial_B}, we get:
\begin{align*}
    \frac{\partial z^{(j)}_{i-1}}{\partial z_{i-1}}\sum^{n_{\bar{\mr u},i-1}}_{k=1} \left({}_k\Bz _{i-1}z_{i-1} + {}_k \Bzct _{i-1}\right)\bar{u}_{i-1,k} &= \sum^{n_{\bar{\mr u},i-1}}_{k=1} \left(\frac{\partial z^{(j)}_{i-1}}{\partial z_{i-1}}{}_k\Bz _{i-1}z_{i-1} +\frac{\partial z^{(j)}_{i-1}}{\partial z_{i-1}} {}_k \Bzct _{i-1}\right)\bar{u}_{i-1,k}\\ &=\sum^{n_{\bar{\mr u},i-1}}_{k=1} \left({}^j_k\Bz _{i-1}z_{i-1}^{(j)} + {}^j_k \Bzct _{i-1}z^{(j-1)}_{i-1}\right)\bar{u}_{i-1,k} 
    \end{align*}
    where ${}^j_k \Bz_{i-1}$ and ${}^j_k \Bzct _{i-1}$ take the form described in \eqref{eq:kron_gradient_Ax} and \eqref{eq:kron_gradient_partial_B}.
Now we can define, 
\begin{equation}\label{eq:B_matrix_blti_nl_proof}
    {}_k \Bz_i = \begin{bmatrix}
        0 & & & & \\
        {}_k \Bzct_{i-1} & {}_k \Bz _{i-1} & & & \\
        & {}^2_k \Bzct_{i-1} & {}^2_k \Bz _{i-1} & & \\
        & & \ddots & \ddots & \\
        & & & {}^{p_i}_k \Bzct_{i-1} & {}{}^{p_i}_k \Bz _{i-1}
    \end{bmatrix}, \quad {}_k\Bzct_i = \begin{bmatrix}
        0 \\ \vdots \\ 0
    \end{bmatrix}\in\mathbb{R}^{n_{\mr z , i}}.
\end{equation}
 As $\bar{u}_i \equiv \bar{u}_{i-1}$ and $\Ap_i$ is the same as in the PITI form, the obtained dynamics are: 
\begin{equation}\label{eq:bilinear_form_in_u_proof}
    \dot{z}_i=\Ap_i z_i + \sum^{n_{\bar{\mr u},i}}_{k=1} \left({}_k\Bz _{i}z_{i} + {}_k \Bzct _{i}\right)\bar{u}_{i,k}
\end{equation}
which is in the form of \eqref{eq:BLTI_representation:state}. 

Regarding the output equation, $\Sigma^{\text{PITI}}_{i-1}$ is bilinear and has no feedtrough by assumption, which means that $\Py _{i-1}(z_{i-1}) \Ry _{i-1} (\bar{u}_{i-1}) =0$. Note that, in  \eqref{eq:expanded_g_ie_decomp}, all input related terms become zero, hence $\Py _{i}(z_{i}) \Ry _{i} (\bar{u}_{i})$ also becomes zero. As such, the output equation of the $\Sigma^{\text{PITI}}_i$ block is described by:
\begin{equation}
    \bar{y}_i = \Cp _i z_i,
\end{equation}
which is linear and has no feedtrough.

\bibliographystyle{siamplain}
\bibliography{siads_block_embedding_references}

\end{document}